\documentclass[10pt, conference]{IEEEtran}
\usepackage[usenames,dvipsnames]{color}
\usepackage{subfigure}
\usepackage{graphics} 
\usepackage{epsfig} 
\usepackage{amsmath} 
\usepackage{amssymb}  
\usepackage{psfrag}
\usepackage{bm}
\usepackage{amsbsy}
\usepackage[justification=centering]{caption}
\usepackage{mdwlist}    
\usepackage{epstopdf} 
\DeclareGraphicsExtensions{.pdf,.eps,.png,.jpg,.mps} 

\makeatletter

\newcommand{\Rmnum}[1]{\expandafter\@slowromancap\romannumeral #1@}

\def\ps@headings{%
\def\@oddhead{\mbox{}\scriptsize\rightmark \hfil \thepage}%
\def\@evenhead{\scriptsize\thepage \hfil \leftmark\mbox{}}%
\def\@oddfoot{}%
\def\@evenfoot{}}
\makeatother
\newtheorem{theorem}{Theorem}

\newtheorem{lemma}{Lemma}

\newtheorem{corollary}{Corollary}
\newtheorem{remark}{Remark}
\newtheorem{assumption}{Assumption}

\def\EE{{\mathbb{E}}}

\newcommand{\II}[1]{\mathcal{I}_{\{#1\}}}

\newcommand{\scale}[2]{#1^{_{(\!#2\!)}}}

\newcommand{\textmathbf}[1]{\boldmath \textbf{#1} \unboldmath}

\newcommand{\nonnegpart}[1]{\left[#1\right]^+}

\begin{document}

\title{Optimal Content Placement for Peer-to-Peer\\Video-on-Demand Systems$^{1}$}
\author{
\IEEEauthorblockN{Bo (Rambo) Tan\\}
\IEEEauthorblockA{Department of Electrical and Computer Engineering\\
University of Illinois at Urbana-Champaign\\
Urbana, IL 61801, USA \\
Email: botan2@illinois.edu}
\and
\IEEEauthorblockN{Laurent Massouli\'e\\}
\IEEEauthorblockA{Technicolor Paris Research Lab\\
Issy-les-Moulineaux Cedex 92648, France\\
Email: laurent.massoulie@technicolor.com}
}

\maketitle

\begin{abstract}
In this paper, we address the problem of content placement in peer-to-peer systems, with the objective of maximizing the utilization of peers' uplink bandwidth resources. We consider system performance under a many-user asymptotic. We distinguish two scenarios, namely ``Distributed Server Networks'' (DSN) for which requests are exogenous to the system, and ``Pure P2P Networks'' (PP2PN) for which requests emanate from the peers themselves. For both scenarios, we consider a {\em loss network} model of performance, and determine asymptotically optimal content placement strategies in the case of a limited content catalogue. We then turn to an alternative ``large catalogue'' scaling where the catalogue size scales with the peer population.
Under this scaling, we establish that storage space per peer must necessarily grow unboundedly if bandwidth utilization is to be maximized. Relating the system performance to properties of a specific random graph model, we then identify a content placement strategy and a request acceptance policy which jointly maximize bandwidth utilization, provided storage space per peer grows unboundedly, although arbitrarily slowly, with system size.
\end{abstract}

\section{Introduction}
\label{sec: intro}
The amount of multimedia \addtocounter{footnote}{1}\footnotetext{Part of the results developed in this paper have made the object of a ``brief announcement'' in \cite{tanmas10_podcba} and further shown in more detail in \cite{tanmas11_infocom}.} 
traffic accessed via the Internet, already of the order of exabytes ($10^{18}$) per month, is expected to grow steadily in the coming years. A peer-to-peer (P2P) architecture, whereby peers contribute resources to support service of such traffic, holds the promise to support such growth more cheaply than by scaling up the size of data centers. More precisely, a large-scale P2P system based on resources of individual users can absorb part of the load that would otherwise need to be served by data centers.

In the present work we address specifically the Video-on-Demand (VoD) application, for which the critical resources at the peers are storage space and uplink bandwidth. Our objective is to ensure that the largest fraction of traffic is supported by the P2P system. More precisely, we look for content placement strategies that enable content downloaders to maximally use the peers' uplink bandwidth, and hence maximally offload the servers in the data centers. Such strategies must adjust to the distinct popularity of video contents, as a more popular content should be replicated more frequently.

We consider the following mode of operation: Video requests are first submitted to the P2P system; if they are accepted, uplink bandwidth is used to serve them at the video streaming rate (potentially via parallel substreams from different peers). They are rejected if their acceptance would require disruption of an ongoing request service. Rejected requests are then handled by the data center. Alternative modes of operation could be envisioned (e.g., enqueueing of requests, service at rates distinct from the streaming rate, joint service by peers and data center,...). However the proposed model is appealing for the following reasons. It ensures zero waiting time for requests, which is desirable for VoD application; analysis is facilitated, since the system can be modeled as a {\em loss network} \cite{kel91}, for which powerful theoretical results are available; and finally, as our results show, simple placement strategies ensure optimal operation in the present model.

In the P2P system we are considering, there are two kinds of peers: boxes and pure users. Their difference is that boxes do contribute resources (storage space and uplink bandwidth) to the system, while pure users do not. This paper focuses on the following two architectures (illustrated in Figure~\ref{fig: architecture}):
\begin{itemize*}
\item \textbf{Distributed Server Network (DSN):} Requests to download contents come only from pure users, and can be regarded as external requests.
\item \textbf{Pure P2P Network (PP2PN):} There are no pure users in the system, and boxes do generate content requests, which can be regarded as ``internal''.
\end{itemize*}

\begin{figure}[t]
\centering
\includegraphics[width=9.0cm]{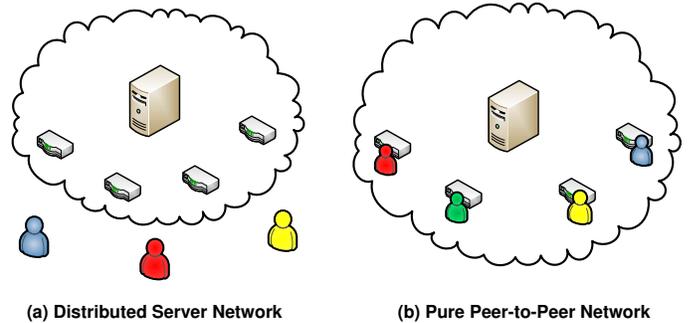}
\caption{Two architectures of P2P VoD systems} \label{fig: architecture}
\end{figure}

The rest of the paper is organized as follows: We review related work in Section~\ref{sec: survey} and introduce our system  model in Section~\ref{sec: model}. 
For the Distributed Server Network scenario, the so-called ``proportional-to-product'' content placement strategy is introduced and shown to be optimal in a large system limit in Section~\ref{sec: content_place_dsn}, where extensive simulation results are also provided. For the Pure P2P Network scenario, a distinct placement strategy is introduced and proved optimal in Section~\ref{sec: content_place_pp2pn}. These results apply for a catalogue of contents of limited size. An alternative model in which catalogue size grows with the user population is introduced in Section~\ref{sec: large_catalog}, where it is shown that the ``proportional-to-product'' placement strategy remains optimal in the DSN scenario in this large catalogue setting, for a suitably modified request management technique.

\section{Related Work}
\label{sec: survey}
The number and location of replicas of distinct content objects in a P2P system have a strong impact on such system's performance. Indeed, together with the strategy for handling incoming requests, they determine whether such requests must either be delayed, or served from an alternative, more expensive source such as a remote data center. Requests which cannot start service at once can either be enqueued (we then speak of a waiting model) or redirected (we then speak of a loss model).

Previous investigations of content placement for P2P VoD systems were conducted by Suh et al. \cite{suhdiokur07}. The problem tackled in~\cite{suhdiokur07} differs from our current perspective, in particular no optimization of placement with respect to content popularity was attempted in this work. Performance analysis of both queueing and loss models are considered in \cite{suhdiokur07}. Valancius et al. \cite{vallaomas09} considered content placement dependent on content popularity, based on a heuristic linear program, and validated this heuristic's performance in a loss model via simulations.

Tewari and Kleinrock \cite{tewkle05PropReplicaPerNode, tewkle06PropReplicaSystem} advocated to tune the number of replicas in proportion to the request rate of the corresponding content, based on a simple queueing formula, for a waiting model, and also from the standpoint of the load on network links. They further established via simulations that Least Recently Used (LRU) storage management policies at peers emulated rather well their proposed allocation.


Wu et al. \cite{WuLi09_keepcache} considered a loss model, and a specific time-slotted mode of operation whereby requests are submitted to randomly selected peers, who accommodate a randomly selected request. They showed that in this setup the optimal cache update strategy can be expressed as a dynamic program. Through experiments, they established that simple mechanisms such as LRU or Least Frequently Used (LFU) perform close to the optimal strategy they had previously characterized.


Kangasharju et al. \cite{KanRosTur07optimizingfile} addressed file replication in an environment where peers are intermittently available, with the aim of maximizing the probability of a requested file being present at an available peer. This differs from our present focus in that the bandwidth limitation of peers is not taken into account, while the emphasis is on their intermittent presence. They established optimality of content replication in proportion to the {\em logarithm} of its popularity, and identified simple heuristics approaching this.




Boufkhad et al. \cite{boufkhad08achievable} considered P2P VoD from yet another viewpoint, looking at the number of contents that can be simultaneously served by a collection of peers.

Content placement problem has also been addressed towards other different optimization objectives. For example, Almeida et al. \cite{almeagver04DeliveryCost} aim at minimizing total delivery cost in the network, and Zhou et al. \cite{ZhouXu07VideoReplica} target jointly maximizing the average encoding bit rate and average number of content replicas as well as minimizing the communication load imbalance of video servers.

Cache dimensioning problem is considered in \cite{Laoutaris03OptStorageAlloc}, where Laoutaris et al. optimized the storage capacity allocation for content distribution networks under a limited total cache storage budget, so as to reduce average fetch distance for the request contents with consideration of load balancing and workload constraints on a given node. Our paper takes a different perspective, focusing on many-user asymptotics so the results show that the finite storage capacity per node is never a bottleneck (even in the ``large catalogue model'', it also scales to infinity more slowly than the system size).


There are obvious similarities between our present objective and the above works. However, none of these identifies explicit content placement strategies at the level of the individual peers, which lead to minimal fraction of redirected (lost) requests in a setup with dynamic arrivals of requests.

Finally, there is a rich literature on loss networks (see in particular Kelly \cite{kel91}); however our present concern of optimizing placement to minimize the amount of rejected traffic in a corresponding loss network appears new.
\section{Model Description}
\label{sec: model}
We now introduce our mathematical model and related notations. Denote the set of all boxes as $\mathcal{B}$. Let $|\mathcal{B}| = B$ and index the boxes from $1$ to $B$. Box $b$ has a local cache $\mathcal{J}_b$ that can store up to $M$ contents, all boxes having the same storage space $M$. We further assume that each box can simultaneously serve $U$ concurrent requests, where $U$ is an integer, i.e., each box has an uplink bandwidth equal to $U$ times the video streaming rate. In particular we assume identical streaming rates for all contents.

The set of available contents is defined as $\mathcal{C}$. Let $|\mathcal{C}|=C$ and index contents from $1$ to $C$. Thus a given box $b$ will be able to serve requests for content $c$ for all $c\in {\mathcal J}_b$.

In a Pure P2P Network, when box $b$ has a request for a certain content $c$, which is coincidentally already in its cache, a ``local service'' is provided and no download service is needed, hence the service to this request consumes no bandwidth resource. The effect of local service on deriving an optimal content placement strategy will be discussed in detail in Section~\ref{sec: content_place_pp2pn}.

In a Distributed Server Network, however, local service will never occur since all the requests are external with respect to the system resources\footnote{In fact the external users issuing requests could  keep local copies of previously accessed content, and hence experience ``local service'' upon re-accessing the same content. But we do not need consider this as this happens outside the perimeter of our system.}.


For a new request that needs a download service, an attempt is made to serve this request by some box holding content $c$, while ensuring that previously accepted requests can themselves be assigned to adequate boxes, given the cache content and bandwidth resources of all boxes. This potentially involves ``repacking'' of requests, i.e., reallocation of all the bandwidth resources in the system (``box-serving-request'' mapping) to accommodate this new download demand pattern. If such repacking can be found, then the request is accepted; otherwise, it is rejected from the P2P system.

It will be useful in the sequel to characterize the concurrent numbers of requests that are amenable to such repacking. Let $\mathbf{n}=\{n_c\}_{c\in{\mathcal C}}$ be the vector of numbers $n_c$ of requests per content $c$. Clearly, a matching of these requests to server boxes is feasible if and only if there exist nonnegative integers $z_{cb}$ (number of concurrent downloads of content $c$ from box $b$) such that
\begin{eqnarray}
\sum_{b: c\in \mathcal{J}_b} z_{cb} & = & n_c,~\forall~c\in \mathcal{C}; \nonumber\\
\sum_{c: c\in \mathcal{J}_b} z_{cb} & \le & U,~\forall~b \in \mathcal{B}. \label{eq: feasible_ori}
\end{eqnarray}
A more compact characterization of feasibility follows by an application of Hall's theorem \cite{bol98} (detailed in Appendix~\ref{sec: proof_hall}), giving that $\mathbf{n}$ is feasible if and only if:
\begin{equation}
\forall~\mathcal{S} \subseteq \mathcal{C},~\sum_{c \in \mathcal{S}} n_c
\le U \left|\{b\in{\mathcal B}:\;  \mathcal{S}\cap  \mathcal{J}_b\ne \emptyset\}\right|. \label{eq: feasible_hall}
\end{equation}

We now introduce statistical assumptions on request arrivals and durations. New requests for content $c$ occur at the instants of a Poisson process with rate $\nu_c$.
We assume that the video streaming rate  is normalized  to $1$, and is the same for all contents. We further assume that all videos have the same duration, again normalized at 1.
Under these assumptions,  the amount of work per time unit brought into the system by content $c$ equals $\nu_c$.

With the above assumptions at hand, assuming fixed cache contents, the vector $\mathbf{n}$ of requests under service is a particular instance of a general stochastic process known as a loss network model. Loss networks were introduced to represent ongoing calls in telephone networks, and exhibit rich structure. In particular, the corresponding stochastic process is reversible, and admits a closed-form stationary distribution. For the Distributed Server Network model, the stationary distribution reads:
\begin{equation}\label{stat_1}
\pi(\mathbf{n})=\frac{1}{Z} \prod_{c\in{\mathcal C}}\frac{\nu_c^{n_c}}{n_c!} \mathcal{I}_{\{\mbox{$\mathbf{n}$ is feasible}\}}.
\end{equation}
In  words, the numbers of requests $n_c$ are independent Poisson random variables with parameter $\nu_c$, conditioned on feasibility of the whole vector $\mathbf{n}$.

Our objective is then to determine content placement strategies so that in the corresponding loss network model, the fraction of rejected requests is minimal. The difficulty in doing this analysis  resides in the fact that the normalizing constant $Z$ is cumbersome to evaluate. Nevertheless, simplifications occur under large system asymptotics, which we will exploit in the next sections.

We conclude this section by the following remark. For simplicity we assumed in the above description that a particular content is either fully replicated at a peer, or not present at all, and that a request is served from only one peer. It should however be noted that we can equally assume that contents are split into sub-units, which can be placed onto distinct peers, and downloaded from such distinct peers in parallel sub-streams in order to satisfy a request. This extension is detailed in Appendix~\ref{sec: parallel}.


\section{Optimal Content Placement in Distributed Server Networks}
\label{sec: content_place_dsn}

We first describe a simple adaptive cache update strategy driven by demand, and show why it converges to a ``predetermined'' content placement called ``proportional-to-product'' strategy. We then establish the optimality of this ``proportional-to-product'' placement in a large system asymptotic regime.

\subsection{The Proportional-to-Product Placement Strategy}
\label{sec: proportional}
A simple method to adaptively update the caches at boxes driven by demand is described as follows:

\noindent\hrulefill\\
\textbf{Demand-Driven Cache Update}

\noindent\hrulefill\\
Whenever a new request comes, with probability $\epsilon B$ ($\epsilon$ is chosen such that $\epsilon B\le 1$), the server picks a box $b$ uniformly at random, and attempts to push content $c$ into this box's cache. If $c$ is already in there, do nothing; otherwise, remove a content selected uniformly at random from the cache.

\noindent\hrulefill\\
%

Since external demands for content $c$ are according to a Poisson process with rate $\nu_c$, we find that under the above simple strategy, content $c$ is pushed at rate $\epsilon \nu_c$ into a particular box which is not caching content $c$. Recall that each box stores $M$ distinct contents, and let $j$ denote a candidate ``cache state'', which is a size $M$ subset of the full content set ${\mathcal C}$. For convenience, let ${\mathcal{J}}$ denote the collection of all such $j$.

With the above strategy, the caches at each box evolve independently according to a continuous-time Markov process. The rate at which cache state $j$ is changed to $j'$, where $j'=j+\{c\}\setminus\{d\}$ for some contents $d\in j$, $c\notin j$, which we denote by $q(j,j')$, is easily seen to be
$
q(j,j')=\epsilon \nu_c/M$.
Indeed, content $d$ is evicted with probability $1/M$, while content $c$ is introduced at rate $\epsilon \nu_c$.

It is easy to verify that the distribution $p(\cdot)$ given by
\begin{equation}
p(j)=\frac{1}{Z}\prod_{c\in j}\nu_c,\quad j\in{\mathcal{J}},
\end{equation}
for some suitable normalizing constant $Z$,
 verifies the follwing equation:
\begin{equation}
p(j)q(j,j')=p(j')q(j',j),\quad j,j'\in{\mathcal{J}}. \label{eq: mc_evolve}
\end{equation}
 The latter relations, known as the local balance equations, readily imply that $p(\cdot)$ is a stationary distribution for the above Markov process; since the process is irreducible, this is the unique stationary distribution.

Thus, we can conclude that under this cache update strategy, the random cache state at any box eventually follows this stationary distribution. This is what we refer to as the \textbf{``proportional-to-product'' placement strategy}, and it is the one we advocate in the Distributed Server Network scenario.\\

\begin{remark}
The customized parameter $\epsilon$ should not be too large, otherwise the burden on the server will be increased due to use of ``push''. Neither should it be too small, otherwise the Markov chain will converge too slowly to the steady state.\hfill$\diamond$
\end{remark}

Under the cache update strategy, the distribution of cache contents needs time to converge to the steady state. However, if we have a priori information about content popularity, we can use a sampling strategy as an alternative way to directly generate proportional-to-product content placement in one go. One  method works as follows: 

\noindent\hrulefill\\
\textbf{Sampling-Based Preallocation}

\noindent\hrulefill\\
Select successively $M$ contents at random in an i.i.d. fashion, according to the probability distribution $\{\hat{\nu}_c\}$, where $\hat{\nu}_c=\nu_c/\sum_{c'\in{\mathcal C}}\nu_{c'}$ is the normalized popularity. If there are duplicate selections of some content, re-run the procedure. It is readily seen that this yields a sample with the desired distribution. 

\noindent\hrulefill\\

An alternative 
sampling strategy which can be faster than the one described above when very popular items are present is given in the Appendix~\ref{sec: analyse_push}.




\subsection{A Loss Network Under Many-User Asymptotics}
\label{sec: loss_network}
We now consider the asymptotic regime called \textbf{``many user--fixed catalogue'' scaling}: The number of boxes $B$ goes to infinity. The system load, defined as
\begin{equation}
\rho \triangleq \frac{\sum_{c\in \mathcal{C}}\nu_c}{BU}, \label{eq: rho}
\end{equation}
is assumed to remain fixed, which is achieved in the present section by assuming that the content collection ${\mathcal C}$ is kept fixed, while the individual rates $\{\nu_c\}$ scale linearly with $B$. We also assume that  the normalized content popularities $\{\hat{\nu}_c\}$ remain fixed as $B$ increases. It thus holds that $\nu_c=\hat{\nu}_c\rho B U$ for all $c\in \mathcal{C}$. Note that although boxes are pure resources rather than users, scaling of $\{\nu_c\}$ with $B$ to infinity actually indicates a ``many-user'' scenario.

To analyze the performance of our proposed proportional-to-product strategy, we require that the cache contents are sampled at random according to this strategy and are subsequently kept fixed. This can either reflect the situation where we use the previously introduced sampling strategy, or alternatively the situation where the cache update strategy has already made the distribution of cache states converge to the steady state, and occurs at a slower time scale than that at which new requests arise and complete.

Note that, as $B$ grows large, the right-hand side in the feasibility constraint (\ref{eq: feasible_hall}) verifies, by the strong law of large numbers,
\begin{equation}
 \left|\{b\in{\mathcal B}:\;  \mathcal{S}\cap  \mathcal{J}_b\ne \emptyset\}\right|\sim B\sum_{j:j\cap S\ne \emptyset} m_j.\label{eq: SLLN}
\end{equation}
Here, $\{m_j\}$ corresponds to a particular content placement strategy, under which 
each box holds a size $M$ content set $j$ with probability $m_j$, and this happens independently over boxes. Specifically, $m_j=\frac{1}{Z}\prod_{c\in j}\hat{\nu}_c$ (where $Z$ is a normalizing constant) corresponds to our proportional-to-product placement strategy.

We now establish a sequence of loss networks indexed by a large parameter $B$. For the $B^{\rm th}$ loss network, requests for content $c \in \mathcal{C}$ (regarded as ``calls of type $c$'') arrive at rate $\scale{\nu_c}{B} = (\rho U \hat{\nu}_c)\cdot B$, each ``virtual link'' $\mathcal{S}\subseteq \mathcal{C}$ has a capacity \begin{equation}
\scale{W_{\mathcal{S}}}{B} \triangleq (U\sum_{j:j\cap S\ne \emptyset} m_j)\cdot B, \label{eq: eq_capacity} 
\end{equation}
and $c\in \mathcal{S}$ represents that virtual link $\mathcal{S}$ is part of the ``route'' which serves call of type $c$.\footnote{Note that this construction in fact admits a form of 
fixed routing which is equivalently transformed from a dynamic routing model where each particular box is regarded as a link and calls of type $c$ can use any single-link route corresponding to a box holding content $c$. This equivalent transform is based on the assumption that repacking is allowed (cf. Section 3.3. in \cite{kel91}). We have already found this equivalent transform by converting feasibility condition \eqref{eq: feasible_ori} to \eqref{eq: feasible_hall} in Section \ref{sec: model}.} 
This particular setup has been identified as the ``large capacity network scaling'' in Kelly \cite{kel91}. There, it is shown that the loss probabilities in the limiting regime where $B\to\infty$ can be characterized via the analysis of an associated variational problem. 

We now describe the corresponding results in \cite{kel91} relevant to our present purpose. For the $B^{\rm th}$ loss network, consider the problem of finding the mode of the stationary distribution (\ref{stat_1}), which corresponds to maximizing $\sum_{c\in \mathcal{C}}(\scale{n_c}{B} \log \scale{\nu_c}{B} - \log \scale{n_c}{B}!)$ over feasible $\scale{\mathbf{n}}{B}$. Then, approximate $\log \scale{n_c}{B}!$ by $\scale{n_c}{B} \log \scale{n_c}{B} - \scale{n_c}{B}$ according to Stirling's formula and replace the integer vector $\scale{\mathbf{n}}{B}$ by a real-valued vector $\scale{\mathbf{x}}{B}$. 
This leads to the following optimization problem:\\\\
{\bf [OPT 1]}
\begin{eqnarray}
\max_{\scale{\mathbf{x}}{B}} & & \sum_{c\in \mathcal{C}}(\scale{x_c}{B} \log \scale{\nu_c}{B} - \scale{x_c}{B} \log \scale{x_c}{B} + \scale{x_c}{B})  \label{opt: mode_primal}\\
s.t.              & & \forall~\mathcal{S} \subseteq \mathcal{C},~\sum_{c \in \mathcal{S}} \scale{x_c}{B} \le \scale{W_{\mathcal{S}}}{B}   \label{eq: primal_constraint} \\
         \hbox{over}         & & \scale{\mathbf{x}}{B}\ge 0.  \nonumber
\end{eqnarray}
The corresponding Lagrangian is given by:
\begin{eqnarray*}
L(\scale{\mathbf{x}}{B},\scale{\mathbf{y}}{B}) &=& \sum_{c\in \mathcal{C}}(\scale{x_c}{B} \log \scale{\nu_c}{B} - \scale{x_c}{B} \log \scale{x_c}{B} + \scale{x_c}{B}) \\
&&+ \sum_{\mathcal{S}\subseteq \mathcal{C}} \scale{y}{B}_{_{\mathcal{S}}}(\scale{W_\mathcal{S}}{B} - \sum_{c \in \mathcal{S}} \scale{x_c}{B}),
\end{eqnarray*}
where $\{\scale{y_{_{\mathcal{S}}}}{B}\}_{_{\mathcal{S}\subseteq \mathcal{C}}}$ are Lagrangian multipliers. The KKT conditions  for this convex optimization problem comprise the original constraints and the following ones:
\begin{equation*}
\scale{\bar{y}}{B}_{_{\mathcal{S}}}(\scale{W_\mathcal{S}}{B} - \sum_{c \in \mathcal{S}} \scale{\bar{x}_c}{B}) = 0,~\scale{\bar{y}}{B}_{_{\mathcal{S}}}\ge 0,~\forall~\mathcal{S}\subseteq \mathcal{C},
\end{equation*}
\begin{equation}
\frac{\partial L(\scale{\bar{\mathbf{x}}}{B},\scale{\bar{\mathbf{y}}}{B})}{\partial \scale{x_c}{B}} = \log \scale{\nu_c}{B} - \log \scale{\bar{x}_c}{B} - \sum_{\mathcal{S}: c \in \mathcal{S}} \scale{\bar{y}}{B}_{_{\mathcal{S}}} = 0,~\forall~c \in \mathcal{C}
\label{eq: Lagrarian_diff}
\end{equation}
where $(\scale{\bar{\mathbf{x}}}{B},\scale{\bar{\mathbf{y}}}{B})$ is a solution to the optimization problem. From equation~(\ref{eq: Lagrarian_diff}), we further get
\begin{equation}
\scale{\bar{x}_c}{B} = \scale{\nu_c}{B} \exp(-\sum_{\mathcal{S}: c \in \mathcal{S}} \scale{\bar{y}}{B}_{_{\mathcal{S}}}),~\forall~ c\in \mathcal{C}.\label{eq: x_opt}
\end{equation}


Then the result that we will need from Kelly \cite{kel91} is the following: 
for the $B^{\rm th}$ loss network, 
the steady state probability of accepting request for $c$, denoted by $\scale{A}{B}_c$, verifies
\begin{equation}
\scale{A}{B}_c = \exp\left(-\sum_{\mathcal{S}: c \in \mathcal{S}} \scale{\bar{y}}{B}_{_{\mathcal{S}}}\right)+O\left(B^{-\frac{1}{2}}\right),~\forall~ c\in \mathcal{C},\label{eq: accept_rate}
\end{equation}
where $\scale{\bar{y}}{B}_{\mathcal S}$ are the Lagrangian multipliers of the previous optimization problem.


\subsection{Optimality of Proportional-to-Product Content Placement}
\label{sec: optimality_prod}
Note that the global acceptance probability, denoted by $A_{sys}$, which also reads $A_{sys}=\sum_{c\in{\mathcal C}}\hat{\nu}_cA_c$, cannot exceed $\min(1,1/\rho)$. Indeed, it is clearly no larger than 1. It cannot exceed $1/\rho$ either, otherwise the system would treat more requests than its available resources.

We now prove that the proportional-to-product content placement not only achieves the optimal global acceptance probability $A_{sys}=\min(1,1/\rho)$, but also achieves fair individual acceptance probabilities, i.e., $A_c=A_{sys}$ for all $c$. More precisely, we have the following theorem:\\

\begin{theorem}
\label{thm: arr_prod_alloc}
By using $m_j = \prod_{c\in j}\hat{\nu}_c /Z$ for all $j\subseteq {\mathcal C}$ s.t. $|j|=M$, where
$Z$ is the normalizing constant, we have $\lim_{B\to\infty}\scale{A}{B}_c =\min\{1, 1/\rho\},~\forall c\in \mathcal{C}$, for fixed  $\rho$ and ${\mathcal C}$.\hfill $\diamond$
\end{theorem}
\

Before giving the proof, we comment on the result. One point to note is that because of \eqref{eq: SLLN}, the above optimal acceptance rate is achieved with probability one under any random sampling which follows the proportional-to-product scheme. Secondly, the optimality of the asymptotic acceptance probability does not depend on $M$, as long as $M\ge 1$. Thus for this particular scaling regime, storage space is not a bottleneck. As we shall see in the next two sections, increasing $M$ {\textbf {does}} improve performance if either local services occur, as in the Pure P2P Network scenario (Section 4), or if the catalogue size $C$ scales with the box population size $B$, a case not covered by the classical literature on loss networks, and to which we turn in Section~\ref{sec: large_catalog_content_placement}.\\

\begin{proof}
First, we consider $\rho \ge 1$. 
Letting 
\begin{equation}
\exp\left(-\sum_{\mathcal{S}: c \in \mathcal{S}} \scale{\bar{y}}{B}_{_{\mathcal{S}}}\right) = 1/\rho,~\forall c\in\mathcal{C}, \label{eq: target_accept}
\end{equation}
we have
\begin{equation}
\forall c\in \mathcal{C},~\sum_{\mathcal{S}: c \in \mathcal{S}} \scale{\bar{y}}{B}_{_{\mathcal{S}}}= \log \rho. \label{eq: y_sum}
\end{equation}
Putting equation (\ref{eq: y_sum}) into (\ref{eq: x_opt}) leads to
\begin{equation}
\forall c\in \mathcal{C},~ \scale{\bar{x}}{B}_c = \scale{\nu}{B}_c /\rho. \nonumber
\end{equation}
Thus, inequality~(\ref{eq: primal_constraint}) in OPT 1 becomes
\begin{equation}
\forall \mathcal{S}\subseteq \mathcal{C},~\sum_{c\in \mathcal{S}} \scale{\nu}{B}_c\le \rho \sum_{j:j \cap \mathcal{S} \not= \emptyset} m_j BU. \label{eq: opt1_feasible2}
\end{equation}
Since $\scale{\nu}{B}_c = \rho BU \cdot \hat{\nu}_c$ and $\sum_{c\in_{\mathcal{C}}} \hat{\nu}_c = 1$, inequality~(\ref{eq: opt1_feasible2}) further becomes, upon explicitly writing out the normalization constant $Z$:
\begin{equation}
\forall \mathcal{S}\subseteq \mathcal{C},~\sum_{c\in\mathcal{S}}\hat{\nu}_c \cdot \!\!\!\!\!\! \sum_{_{\mathcal{G}:~\mathcal{G}\subseteq\mathcal{C}\atop {~~~~|\mathcal{G}|=M}}} \prod_{c\in \mathcal{G}}\hat{\nu}_c
\le \sum_{c\in\mathcal{C}}\hat{\nu}_c \cdot \!\!\!\!\!\!\sum_{_{\mathcal{G}:~ \mathcal{G}\cap\mathcal{S}\not=\emptyset \atop {{\mathcal{G}\subseteq\mathcal{C}}\atop{~|\mathcal{G}|=M}}}} \prod_{c\in \mathcal{G}}\hat{\nu}_c. \label{eq: feasible_prod}
\end{equation}
Two types of product terms (mapped to subsets $\mathcal{K}\subseteq \mathcal{C}$) appear on both sides:
\begin{enumerate*}
\item[\Rmnum{1}.] $\prod_{c\in \mathcal{K}}\hat{\nu}_c$: $|\mathcal{K}| = M+1,~\mathcal{K}\cap S\not=\emptyset$.
\item[\Rmnum{2}.] $(\prod_{c\in \mathcal{K}}\hat{\nu}_c)\cdot \hat{\nu}_{c'}$: $c'\in \mathcal{K}\cap S,~|\mathcal{K}| = M$.
\end{enumerate*}
To show whether inequality~(\ref{eq: feasible_prod}) hold, we only have to prove that given any $\mathcal{S}\subseteq\mathcal{C}$, for each product term (related to a $\mathcal{K}$) which appears in one inequality corresponding to a certain $\mathcal{S}$, its multiplicity on the left hand side is no more than that on the right hand side.
\begin{enumerate*}
\item {\bf For a product term of Type \Rmnum{1}:}
    \begin{itemize*}
    \item On the LHS: Since $\prod_{c\in\mathcal{K}}\hat{\nu}_c = \prod_{c\in\mathcal{G}}\hat{\nu}_c \cdot \hat{\nu}_{c'}$ for some $\mathcal{G}\subseteq \mathcal{C}$ and $c'\in \mathcal{S}\cap \mathcal{K}$, where $\mathcal{G}$ is a size $M$ content set, $c'\not\in \mathcal{G}$, and $\mathcal{K} = \mathcal{G}+\{c'\}$. It is easy to see that we have $|\mathcal{S}\cap \mathcal{K}|$ different choice of $c'$ in a $\mathcal{K}$, so the multiplicity of this product term on the LHS equals $|\mathcal{S}\cap \mathcal{K}|$.
    \item On the RHS: When $|\mathcal{S}\cap \mathcal{K}| \ge 2$, for any $c'\in \mathcal{K}$, $\mathcal{K}\setminus \{c'\}$ is a size $M$ content set of which the intersect with $\mathcal{S}$ is not empty, hence the multiplicity equals $|\mathcal{K}|~(=M+1)$. When $|\mathcal{S}\cap \mathcal{K}| = 1$, the exception to the above case is that if $c'\in \mathcal{S}\cap \mathcal{K}$, then $\mathcal{K}\setminus \{c'\}$ is a size $M$ content set which has no intersect with $\mathcal{S}$ and is actually impossible to appear in the second summation term (over all size $M$ content sets $\mathcal{G}$ s.t. $\mathcal{G}\cap \mathcal{S}\not=\emptyset$) in inequality~(\ref{eq: feasible_prod}). Thus, the multiplicity equals $|\mathcal{K}|-1~(=M)$.
    \end{itemize*}
    From above, we can see that the multiplicity of the product term on the LHS is always no more than that on the RHS.

\item {\bf For a product term of Type \Rmnum{2}:}\\
$\mathcal{K}$ is actually already a size $M$ content set $\mathcal{G}$ s.t. $\mathcal{G}\cap\mathcal{C}\not=\emptyset$. Therefore, it is easy to see that on both sides, the multiplicities of this product term are both $1$.
\end{enumerate*}
Now we can conclude that inequality~(\ref{eq: feasible_prod}) holds for all $\mathcal{S}\subseteq \mathcal{C}$, and continue to check the complementary slackness. Given $\rho \ge 1$,
one simple solution to equation~(\ref{eq: y_sum}) reads:
\begin{equation}
\forall~\mathcal{S}\subseteq \mathcal{C},~\scale{\bar{y}}{B}_{_{\mathcal{S}}} = \log \rho \cdot \mathcal{I}_{_{\{\mathcal{S}=\mathcal{C}\}}}.  \label{eq: sol_multiplier}
\end{equation}
Besides, inequality~(\ref{eq: feasible_prod}) is tight for $\mathcal{S}=\mathcal{C}$ (we even do not need to check this when $\rho =1$). Therefore, complementary slackness is always satisfied with solution~(\ref{eq: sol_multiplier}).

So far we have proved that the KKT condition holds when $\rho\ge 1$. When $\rho <1$, we modify \eqref{eq: target_accept} by letting 
\begin{equation}
\exp\left(-\sum_{\mathcal{S}: c \in \mathcal{S}} \scale{\bar{y}}{B}_{_{\mathcal{S}}}\right) = 1,~\forall c\in\mathcal{C}, \label{eq: target_accept2}
\end{equation}
and hence there is an additional factor $1/\rho > 1$ on the RHS of inequality~(\ref{eq: feasible_prod}). Since the old version of inequalities~(\ref{eq: feasible_prod}) is proved to hold, the new version automatically holds, but none of them is tight now. However, 
from \eqref{eq: target_accept2} we have $\scale{\bar{y}}{B}_{_{\mathcal{S}}} = 0,~\forall~\mathcal{S}\subseteq \mathcal{C}$, which means complementary slackness is always satisfied (similar to $\rho =1$). 

Therefore, according to equation \eqref{eq: accept_rate}, it can be  concluded that by using $m_j = \prod_{c\in j}\hat{\nu}_c/Z$ for all $j$, we can achieve 
\begin{equation*}
\scale{A}{B}_c =\min\{1,1/\rho\}+O\left(B^{-\frac{1}{2}}\right),~\forall c\in \mathcal{C},
\end{equation*}
so $\lim_{B\rightarrow \infty} \scale{A}{B}_c = \min\{1,1/\rho\}$. 
\end{proof}

\subsection{Simulation Results}
In this subsection, we use extensive simulations to evaluate the performances of the two implementable schemes proposed in Subsection \ref{sec: proportional} which follow the ``proportional-to-product'' placement strategy, namely the sampling-based preallocation scheme and the demand-driven cache update (labeled as \textbf{``SAMP''} and \textbf{``CU''}, respectively).

We compare the results with the theoretical optimum (i.e., loss rate for each content equals $(1-1/\rho)^+$; the curves are labeled as \textbf{``Optimal''}) and a uniform placement strategy (labeled as \textbf{``UNIF''}) defined as the following:  
first, permute all the contents uniformly at random, resulting in a content sequence $\{c_i\}$, for $1\le i \le C$; then, 
push the $M$ contents indexed by subsequence $\{c_{(j\mod C)}\}_{bM+1 \le j \le (b+1)M}$ into the cache of box $b$, for $1\le b\le B$.
UNIF is also used to generate the initial content placement for CU so that the loss rate can be reduced during the warm-up period.

If not further specified, the default parameter setting is as follows:
The popularity of contents $\{\hat{\nu}_c\}$ follows a zipf-like distribution (see e.g. \cite{brecaofan99}), i.e., 
\begin{equation}
\hat{\nu}_c = \frac{(c_0 + c)^{-\alpha}}{\sum_{c'\in \mathcal{C}} (c_0 + c')^{-\alpha}}, \label{eq: zipf}
\end{equation}
with a decaying factor $\alpha > 0$ and the shift $c_0 \ge 0$. We use  $\alpha = 0.8$ and $c_0 = 0$.
The content catalogue size $C = 500$ and the number of boxes $B = 4000$. Each box can store $M = 10$ contents and serve at most $U = 4$ concurrent requests. The duration of downloading each content is  exponentially distributed with mean equal to $1$ time unit. The parameter $\epsilon$ in the cache update algorithm is set as $1/B$ such that upon a request, one box will definitely be chosen for cache update.  

For every  algorithm, we take the average over $10$ independent repetitive experiments, each of which is observed for $10$ time units. According to the sample path, the initial $1/5$ of the whole period is regarded as a ``warm-up'' period and hence ignored in the calculation of final statistics.\footnote{We can get enough samples during each observation period of 10 time units (for example, when $\rho=1$, $B=4000$ and $U=4$, the average arrivals would be $160000$). It has also been checked that after the warm-up period, the distribution of cache states well approximates the proportional-to-product placement and is kept quite stably for the remaining observation period.}  

Some implementation details are not captured by our theoretical model, but should be considered in simulations. Upon a request arrival, the most idle box (i.e., with the largest number of free connections) among all the boxes which hold the requested content is chosen to provide the service, for the purpose of load balancing. If none of them is idle, we use a heuristic repacking algorithm which iteratively reallocates  the ongoing services among boxes, in order to handle as many  requests as possible while still respects load balancing. One important parameter which trades off the repacking complexity and the performance is the maximum number of iterations $t_r^{max}$, which is set as ``undefined'' by default (i.e., the iterations will continue until the algorithm terminates; theoretically there are at most $C$ iterations). Other details regarding the repacking algorithm can be found in Appendix \ref{sec: implement}. We will see an interesting observation about $t_r^{max}$ later. 

\begin{figure}[t]
\centering
\includegraphics[width=9.2cm]{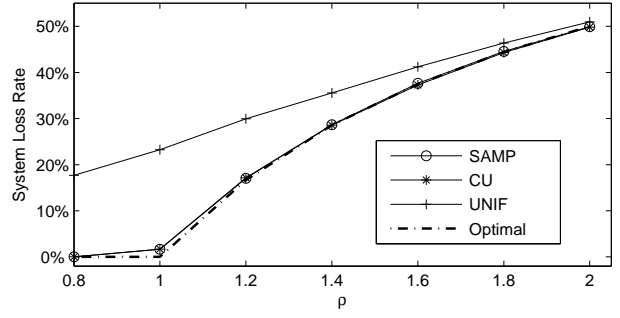}
\caption{System loss rates under different traffic loads} \label{fig: dsn_loss_rho}
\end{figure}

Figure~\ref{fig: dsn_loss_rho} evaluates system loss rates under different traffic loads $\rho$. Our two algorithms SAMP and CU, which target the proportional-to-product placement, both match the theoretically optimum very well.\footnote{In fact, around $\rho = 1$, they perform a little worse than the optimum. The reason is that $\rho=1$ is the ``critical traffic load'' (a separation point between zero-loss and nonzero-loss ranges), under which the simulation results are easier to incur deviation from the theoretical value.} On the other hand, the UNIF algorithm, which does not utilize any information about content popularity, incurs a large loss even if the system is underloaded ($\rho < 1$). 
The gain of proportional-to-product placement over UNIF becomes less significant as the traffic load grows, which can be easily expected.

\begin{figure}[t]
\centering
\includegraphics[width=9.2cm]{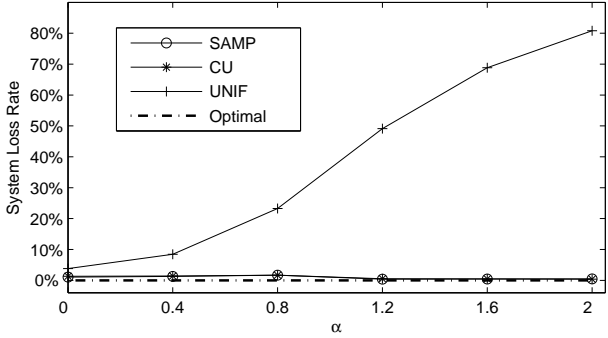}
\caption{System loss rates with different $\alpha$ ($\rho=1$)} \label{fig: dsn_loss_alpha}
\end{figure}

In Figure~\ref{fig: dsn_loss_alpha}, when the decaying factor $\alpha$ in the zipf-like distribution increases, the distribution of placed contents generated by UNIF has a higher discrepancy from the real content popularity distribution, so UNIF performs worse. On the other hand, the two proportional-to-product strategies are insensitive to the change of content popularity, as we expected.

\begin{figure}[t]
\centering
\includegraphics[width=9.2cm]{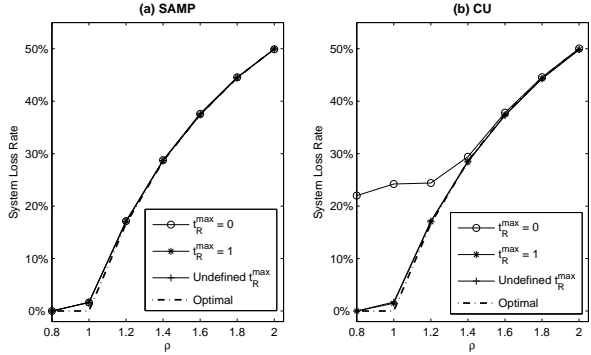}
\caption{Effect of repacking on the system loss rate} \label{fig: dsn_loss_rpk}
\end{figure}

Figure~\ref{fig: dsn_loss_rpk} shows the effect of repacking on the system loss rate. In sub-figure (a), we find that under SAMP, repacking is not  necessary. In sub-figure (b) which shows the performances of CU, when $\rho$ is low, one iteration of repacking is sufficient to make the performance close enough to the optimum; when $\rho$ is high, repacking also becomes unnecessary. The main take-away message from this figure is that we can execute a repacking procedure of very small complexity   
without sacrificing much performance. The reason is that when the server picks a box to serve a request, it already respects the rule of load balancing.

We then explain why CU still needs one iteration of repacking to improve the performance when $\rho$ is low. Note that during the cache update, it is possible that the box is currently uploading the ``to-be-kicked-out'' content to some users. If repacking is enabled, those ongoing services can be repacked to other boxes (see details in Appendix \ref{sec: implement}), but if $t_r^{max} = 0$ (no repacking), they will be terminated and counted as losses. When $\rho$ is high, however, boxes are more likely to be busy, which leads to the failure of repacking, so repacking makes no difference. 

\begin{figure}[t]
\centering
\includegraphics[width=9.2cm]{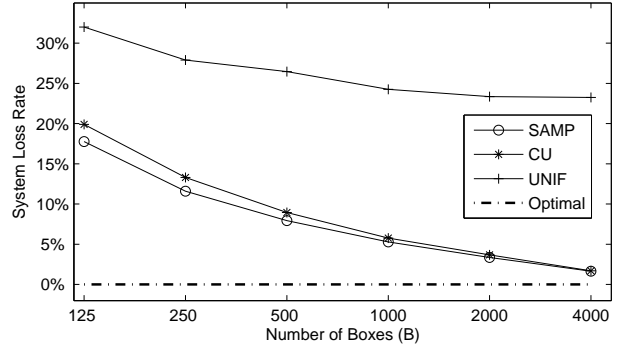}
\caption{System loss rates with different number of boxes} \label{fig: dsn_loss_B}
\end{figure}

\begin{figure}[t]
\centering
\includegraphics[width=9.2cm]{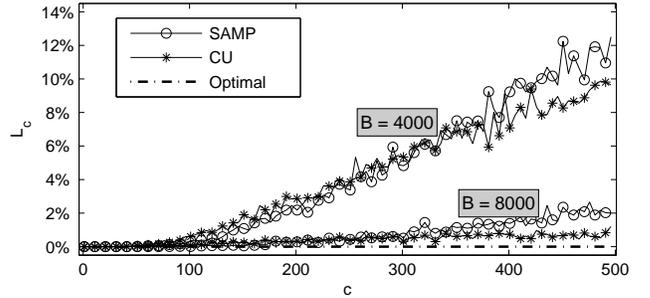}
\caption{Loss rate of requests for each content ($\rho = 1$)} \label{fig: dsn_loss_homo_B}
\end{figure}

Recall that the proportional-to-product placement is only optimal when the number of boxes $B\rightarrow\infty$. Figures~\ref{fig: dsn_loss_B} and \ref{fig: dsn_loss_homo_B} then show the impact of a finite $B$. In Figure~\ref{fig: dsn_loss_B}, as $B$ decreases, the system loss rate of every algorithms increases (compared to the two proportional-to-product strategies, UNIF is less sensitive to $B$). In Figure~\ref{fig: dsn_loss_homo_B}, non-homogeneity in the individual loss rates of requests for each content also reflects a deviation from the theoretical result (when $B\rightarrow \infty$, the loss rates of the requests for all the contents are proved to be identical). As expected, increasing the number of boxes (from $4000$ to $8000$) makes the system closer to the limiting scenario and the individual loss rates more homogeneous. Another observation is that as the popularity of a content decreases (in the figure, the contents are indexed in the descending order of their popularity), the individual loss rate increases. However, according to Figure \ref{fig: dsn_loss_rho}, those less popular contents do not affect the system loss rate much even if they incur high loss, since their weights $\{\hat{\nu}_c\}$ are also lower. 

In fact, if we choose a smaller content catalogue size $C$ or a larger cache size $M$, simulations  show the negative impact of a finite $B$ will be reduced (the figures are omitted here). This tells us that if $C$ scales with $B$ rather than being fixed, the proof of optimality under the loss network framework in Subsection \ref{sec: loss_network} is no longer valid and $M$ must be a bottleneck against the performance of the optimal algorithm. We will solve this problem by introducing a certain type of ``large catalogue model'' later in Section \ref{sec: large_catalog}.

\section{Optimal Content Placement in Pure Peer-to-Peer Networks}
\label{sec: content_place_pp2pn}
In the Pure P2P Network scenario, when box $b$ has a request for content $c$ which is currently in its own cache, a ``local service'' will be provided and no download bandwidth in the network will be consumed. 
To simplify our analysis, each request for a specific content is assumed to originate from a box chosen uniformly at random (this in particular assumes identical tastes of all users).

This means that the effective arrival rate of the requests for content $c$ which generates traffic load actually equals $\tilde{\nu}_c \triangleq \nu_c (1-\tilde{m}_c)$, where $\tilde{m}_c$ is defined as the fraction of boxes who have cached content $c$.
Let $\rho_c \triangleq \rho \hat{\nu}_c$ denote the traffic load generated by requests for content $c$, and $\lambda_c$ denote the  fraction of the system bandwidth resources used to  serve  requests for content $c$. Obviously, $\sum_{c\in\mathcal{C}}\lambda_c \le 1$.
The traffic load absorbed by the P2P system either via local services or via service from another box is then upper-bounded by
\begin{equation}
\tilde{\rho} = \sum_{c\in\mathcal{C}} \rho_c \tilde{m}_c + \left[\rho_c (1-\tilde{m}_c)\right]\wedge \lambda_c, \label{eq: thinned_traffic_load}
\end{equation}
where ``$\wedge$'' denotes the minimum operator.

We will use this simple upper bound to identify an optimal placement strategy in the present Pure P2P Network scenario. To this end, we shall establish  that our candidate placement strategy asymptotically achieves this performance bound, namely absorbs a portion $\tilde{\rho}$ in the limit where $B$ tends to infinity.

To find the optimal strategy, we introduce a variable $x_c \triangleq \left[\rho_c (1-\tilde{m}_c)\right] \wedge \lambda_c$ for all $c$. Note further that the fraction $\lambda_c$ is necessarily bounded from above by $\tilde{m}_c$, as only those boxes holding $c$ can devote their bandwidth to serving $c$.
 It is then easy to see that the quantity $\tilde{\rho}$ in (\ref{eq: thinned_traffic_load}) is no larger than the optimal value of the following  linear programming problem:\\\\
\textbf{[OPT 2]}
\begin{eqnarray*}
\max_{\mathbf{\tilde{m}},\bm{\lambda}, \mathbf{x}} & & \sum_{c\in\mathcal{C}} (\rho_c \tilde{m}_c + x_c) \\
s.t.              & & \forall~c\in \mathcal{C},~0\le \tilde{m}_c \le 1,~0\le \lambda_c \le \tilde{m}_c;\\
                  & & \forall~c\in \mathcal{C},~0\le x_c\le \lambda_c,~x_c \le \rho_c(1-\tilde{m}_c);\\
                  & & \sum_{c\in\mathcal{C}} \tilde{m}_c = M,~\sum_{c\in\mathcal{C}} \lambda_c \le 1.
\end{eqnarray*}

The following theorem gives the structure of an optimal solution to OPT 2, and as a result suggests an optimal placement strategy.\\

\begin{theorem}
\label{thm: wf_with_reservation}
Assume that $\{\hat{\nu}_c\}$ are ranked in descending order. The following solution solves OPT 2:
\begin{itemize*}
\item For $1\le c \le M-1$, $\tilde{m}_c = 1$, $\lambda_c = x_c = 0$.
\item For $M \le c \le c^*$, $\tilde{m}_c = \lambda_c = x_c = \rho_c /(1+\rho_c)$, where $c^*$ satisfies that
        $$\sum_{c=M}^{c^*} \frac{\rho_c}{1+\rho_c} \le 1,~\mbox{but}~\sum_{c=M}^{c^*+1} \frac{\rho_c}{1+\rho_c} > 1.$$
\item For $c=c^*+1$, $\tilde{m}_c = \lambda_c = x_c = 1 - \sum_{c=M}^{c^*} \tilde{m}_c$.
\item For $c^*+2 \le c \le C$, $\tilde{m}_c = \lambda_c = x_c = 0$.\hfill $\diamond$
\end{itemize*}
\end{theorem}

The proof consists in checking that the KKT conditions are met for the above candidate solution. Details are given in Appendix \ref{sec: proof_wf_reserve}.

The above optimal solution suggests the following placement strategy:

\noindent\hrulefill\\
\textbf{``Hot-Warm-Cold'' Content Placement Strategy}

\noindent\hrulefill\\
Divide the contents into three different classes according to their popularity ranking (in  descending order):
\begin{itemize*}
\item \textbf{Hot:} The $M-1$ most popular contents. At each box, $M-1$ cache slots are reserved for them to make sure that requests for these contents are always met via local service.
\item \textbf{Warm:} The contents with indices from $M$ to $c^*+1$ (or $c^*$ if $\sum_{c=M}^{c^*}\tilde{m}_c = 1$). For these contents, a fraction $\tilde{m}_c$ of all the boxes will store content $c$ in their remaining one cache slots, where the value of $\tilde{m}_c$ is given in Theorem~\ref{thm: wf_with_reservation}. All requests for these contents (except $c^*+1$ if it is classified as ``warm'') can be served, at the expense of all bandwidth resources.
\item \textbf{Cold:} The other less popular contents are not cached at all.
\end{itemize*}

\noindent\hrulefill \vspace{0.1in}

\begin{remark}
The requests for the $c^*$ most popular contents (``hot'' contents and ``warm'' contents except content $c^*+1$) incur zero loss, while the requests for the $C-c^*-1$ least popular contents incur $100\%$ loss. There is a partial loss in the requests for content $c^*+1$ if $\sum_{c=M}^{c^*}\tilde{m}_c < 1$.

Note that the placement for ``warm'' contents looks like the ``water-filling'' solution in the problem of allocating transmission powers onto different OFDM channels to maximize the overall achievable channel capacity in the context of wireless communications \cite{TseVis_Wireless05}.\hfill $\diamond$
\end{remark}

\

Under this placement strategy, the maximum upper bound on the absorbed traffic load reads
\begin{equation*}\label{eq: absorbed_traffic_ub}
\tilde{\rho} = \sum_{c=1}^{c^*}\rho_c + (\rho_{c^*+1}+1)\left(1 - \sum_{c=M}^{c^*} \frac{\rho_c}{1+\rho_c}\right).
\end{equation*}
We then have the following corollary:
\begin{corollary}
Considering the large system limit $B\to\infty$, with fixed catalogue and associated normalized popularities $\{\hat{\nu_c}\}$ as considered in Subsection~\ref{sec: loss_network}, the proposed ``hot-warm-cold'' placement strategy achieves an asymptotic fraction of absorbed load equal to the above upper bound $\tilde{\rho}$, and is hence optimal in this sense.\hfill$\diamond$
\end{corollary}

\

\begin{proof}
With the proposed placement strategy, hot (respectively, cold) contents never trigger accepted requests, since all incoming requests are handled by local service (respectively, rejected). For warm contents, because each box holds only one warm content, it can only handle requests for that particular warm content. As a result, the processes of ongoing requests for distinct warm contents evolve independently of one another. For a given warm content $c$, the corresponding number of ongoing requests behaves as a simple one-dimensional loss network with arrival rate $\nu_c(1-\tilde{m}_c)$ and service capacity $\tilde{m}_c BU$. For $c=M,\ldots,c^*,$ one has $\tilde{m}_c=\rho_c/(1+\rho_c)$ where $\rho_c=\nu_c/(BU)$, so both the arrival rate and the capacity of the corresponding loss network equal $\tilde{m}_c BU$. The asymptotic acceptance probability as $B\to\infty$ then converges to $1$ and the accepted load due to both local service and services from other boxes converges to $\rho_c$. For content $c^*+1$ (if $\tilde{m}_{c^*+1}>0$), the corresponding loss network has arrival rate $\nu_{c^*+1}(1-\tilde{m}_{c^*+1})$ and service capacity $ \tilde{m}_{c^*+1}BU$. Then, in the limit $B\to\infty$, the accepted load (due to both local services and services from other boxes) reads $\rho_{c^*+1}\tilde{m}_{c^*+1}+\tilde{m}_{c^*+1}$ (which is actually smaller than $\rho_{c^*+1}$). Summing the accepted loads of all contents yields the result.
\end{proof}

\section{Large Catalogue Model}
\label{sec: large_catalog}
Keeping the many-user asymptotic, we now consider an alternative model of content catalogue, which we term the ``large catalogue'' scenario. The set of contents $\mathcal{C}$ is divided into a fixed number of ``content classes'',  indexed by $i\in{\mathcal I}$. In class $i$, all the contents have the same popularity (arrival rate) $\nu_i$. The number of contents within class $i$ is assumed to scale in proportion to the number of boxes $B$, i.e., class $i$ contains $\alpha_i B$ contents for some fixed scaling factor $\alpha_i$. We further define $\alpha \triangleq \sum_i \alpha_i$.
With the above assumptions, the system traffic load $\rho$ in equation~(\ref{eq: rho}) reads
\begin{equation}\label{load-large_catalogue}
\rho=\frac{1}{U}\sum_{i\in \mathcal{I}}\alpha_i \nu_i.
\end{equation}
The primary motivation for this model is mathematical convenience: by limiting the number of popularity values we limit the ``dimensionality'' of the request distribution, even though we now allow for a growing number of contents. It can also be justified as an approximation, that would result from batching into a single class all contents with a comparable popularity. Such classes can also capture the movie type (e.g. thriller, comedy) and age (assuming popularity decreases with content age).

We use $\hat{\upsilon}_i$ to denote the normalized popularity  of content class $i\in \mathcal{I}$ and it reads $\sum_{i\in \mathcal{I}}\hat{\upsilon}_i = 1$. It is reasonable to regard each $\hat{\upsilon}_i$ as fixed. $\hat{\nu}_i \triangleq \hat{\upsilon}_i / (\alpha_i B)$ represents the normalized popularity of a specific content in class $i$, which decreases as the number of contents in this class $\alpha_i B$ increases, since users now have more choices within each class. In practice, an online video provider company which uses the Distributed Server Network architecture adds both boxes and available movies of each type to attract more user traffic, under a constraint of a maximum tolerable traffic load $\rho$.

Returning to the Distributed Server Network model of Section \ref{sec: content_place_dsn}, we consider the following questions: What amount of storage is required to ensure that memory space is not a bottleneck? Is the proportional-to-product placement strategy still optimal under the large-catalogue scaling?

\subsection{Necessity of Unbounded Storage}
We first establish that bounded storage will strictly constrain utilization of bandwidth resources. To this end we need the following lemma:\\

\begin{lemma}\label{lemma: necessity_unbounded_storage}
Consider the system under large catalogue scaling, with fixed weights $\alpha_i$ and cache size $M$ per box. Define $M'\triangleq \lceil 2M/\alpha\rceil$. Then

(i) More than half of the contents are replicated at most $M'$ times, and

(ii) For each of these contents, the loss probability is at least $E(\inf_i \nu_i, M' U)>0$, where $E(\cdot,\cdot)$ is the Erlang function \cite{kel91} defined as:
    $$E(\nu,C) \triangleq\frac{\nu^C}{C!}\left[\sum_{n=1}^C \frac{\nu^n}{n!}\right]^{-1}.$$
    \hfill $\diamond$
\end{lemma}

\begin{proof}
We first prove part (i). Note that the total number of content replicas in the system equals $B\!M$. Thus, denoting by $f$ the fraction of contents replicated at least $M'+1$ times, it follows that $f \alpha B (M'+1)\le BM$, which in turn yields
$$f \le \frac{M}{\alpha\left(\lceil 2M/\alpha\rceil + 1\right)} \le \frac{M}{2M+\alpha} < \frac{1}{2},$$
which implies statement (i).

To prove part (ii), we establish the following general property for a loss network (equivalent to our original system) with call types $j\in{\mathcal J}$, corresponding arrival rates $\nu_j$, and  capacity (maximal number of competing calls) $C_l$ on link $\ell$ for all $\ell\in \mathcal{L}$. We use $\ell\in j$ to indicate that the route for calls of type $j$ comprises link $\ell$. Denoting the loss probability of calls of type $j$ in such a loss network as $p_j$, we then want to prove
\begin{equation}\label{gen_bound}
p_j
\ge E(\nu_j,C'_j),
\end{equation}
where $C'_j \triangleq \min_{\ell\in j} C_{\ell}$, i.e., the capacity of the bottleneck link on the route for calls of type $j$.

Note that the RHS of the above inequality is actually the loss probability of a loss network with only calls of type $j$ and capacity $C'_j$. Fixing index $j$, we define this loss network as an auxiliary system and consider the following coupling construction which allows us to deduce inequality \eqref{gen_bound}: Let $X_k$ be the number of active calls of type $k$ in the original system for all $k$, and let $X'_j$ denote the number of active calls of type $j$ in the auxiliary system. Initially, $X_j(0) = X'_j(0)$. The non-zero transition rates for the joint process $(\{X_k\}_{k\in K},X'_j)$ are given by
$$
\begin{array}{ll}
k\ne j:\; X_k\to X_k+1&\!\!\!\!\hbox{at rate }\displaystyle\nu_k \prod_{\ell \in j}\II{\sum_{k \ni \ell}X_k<C_{\ell}},\\
k\ne j:\; X_k\to X_k-1&\hbox{at rate }X_k,\\
(X_j,X'_j)\to(X_j+1,X'_j+1)&\hbox{at rate } \nu_j^{both},\\
(X_j,X'_j)\to(X_j+1,X'_j)&\hbox{at rate }\nu_j^{ori},\\
(X_j,X'_j)\to(X_j,X'_j+1)&\hbox{at rate }\nu_j^{a
ux},\\
(X_j,X'_j)\to(X_j-1,X'_j-1)&\hbox{at rate }X_j,\\
(X_j,X'_j)\to(X_j,X'_j-1)&\hbox{at rate }\nonnegpart{X'_j-X_j},
\end{array}
$$
where 
\begin{eqnarray*}
\nu_j^{both} &\triangleq & \nu_j \II{X'_j<C'_j}\cdot\prod_{\ell \in j} \II{\sum_{k \ni \ell}X_k<C_{\ell}},\\
\nu_j^{ori} &\triangleq & \nu_j \II{X'_j=C'_j}\cdot\prod_{\ell \in j} \II{\sum_{k \ni \ell}X_k<C_{\ell}},\\
\nu_j^{aux} &\triangleq & \nu_j \II{X'_j<C'_j}\cdot\II{\exists \ell \in j \mbox{ s.t. } \sum_{k\in \ell}X_k=C_{\ell}}.
\end{eqnarray*}
It follows from Theorem 8.4 
in \cite{DraMas_epidemics10} that $\{X_k\}$ is indeed a loss network process with the original dynamics, and that $X'_j$ is a one-dimensional loss network with capacity $C'_j$ and arrival rate $\nu_j$. From the construction, we can see that all transitions preserve the inequality $X_j(t)\le X'_j(t)$ for all $t\ge 0$, due to the following reason: Once $X_j$ increases by 1, $X'_j$ either increases by 1 or equals the capacity limit $C'_j$, and for the latter case, the  corresponding transition rate $\nu_j^{ori}$ implies that $X_j \le C'_j = X'_j$. Similarly, once $X'_j$ decreases by 1, either $X_j$ also decreases by 1, or in the case that $X_j$ does not decrease, it must be that the transition rate $X'_j-X_j$ is strictly positive. In any case, the above inequality is preserved. 

We further let $A_j(t)$, $A'_j(t)$ denote the number of type $j$ external calls, $L_j(t)$, $L_j'(t)$ the number of type $j$ call rejections, and $D_j(t)$, $D_j'(t)$ the number of type $j$ call completions, respectively in the original and auxiliary systems, during time interval $[0,t]$. It follows from our construction that whenever the service for a call of type $j$ completes in the original system, the service for a call of type $j$ also completes in the auxiliary system, hence $D_j(t)\le D'_j(t)$ for all $t\ge 0$. Since $X_j(t) = A_j(t) - D_j(t) - L_j(t)$, $X'_j(t) = A'_j(t) - D'_j(t) - L'_j(t)$ and $A_j(t) = A'_j(t)$, we have $L_j(t) \ge L_j'(t)$. Upon dividing this inequality by $A(t)$ and letting $t$ tend to infinity, one retrieves the announced inequality (\ref{gen_bound}) by the ergodic theorem. 

Back to the context of our P2P system, for those contents which are replicated at most $M'$ times (i.e., the contents considered in part (i)), the rejection rate of content $c$ of type $j$ reads $p_j \ge E(\inf_i\nu_i, C'_j) \ge E(\inf_i\nu_i, M'U).$
\end{proof}
\

The above lemma readily implies the following corollary:
\begin{corollary}
\label{col_1}
Under the assumptions in Lemma~\ref{lemma: necessity_unbounded_storage}, 
The overall rejection probability is at least $\frac{1}{2}E(\min_i \nu_i, M' U)$. 
Indeed, for bounded $M$, $M'$ is also bounded, and $E(\min_i \nu_i,M' U)$ is bounded away from $0$.  \hfill$\diamond$
\end{corollary}


Thus, even when the system load $\rho$ is strictly less than 1, with bounded $M$ there is a non-vanishing fraction of rejected requests, hence a suboptimal use of bandwidth.

\subsection{Efficiency of Proportional-to-Product Placement}
\label{sec: large_catalog_content_placement}
We consider the following \textbf{``Modified Proportional-to-Product Placement"}: Each of the $M$ storage slots at a given box $b$ contains a randomly chosen content. The probability of selecting one particular content $c$ is $\nu_i/(\rho BU)$ if it belongs to class $i$.
In addition, we assume that the selections for all such $M\!B$ storage slots are done independently of one another.\\

\begin{remark}
This content placement strategy can be viewed as a ``balls-and-bins'' experiment. All the $M\!B$ cache slots in the system are regarded as balls, and all the $|\mathcal{C}|$ ($=\sum_{i}\alpha_i B$) contents are regarded as bins. We throw each of the $MB$ balls at random among all the $|\mathcal{C}|$ bins. Bin $c$ (corresponding to content $c$ which belongs to class $i$) will be chosen with probability $\nu_i/(\rho BU)$. Alternatively, the resulting allocation can be viewed as a bipartite  random graph connecting boxes to contents. \hfill{$\diamond$}
\end{remark}
\

Note that this strategy differs from the ``proportional-to-product'' placement strategy proposed in Section~\ref{sec: content_place_dsn}, in that it allows for multiple copies of the same content at the same box. However, by the birthday paradox, we can prove the following lemma which shows that up to a negligible fraction of boxes, the above content placement does coincide with the proportional-to-product strategy.\\

\begin{lemma}
\label{lemma: suff_for_diff_replica}
By using the above content placement strategy, at a certain box, if $M\ll  \sqrt{(\min_i \alpha_i)B}$,
\begin{equation}
\Pr(\mbox{all the $M$ cached contents are different}) \approx 1. \label{eq: almost_all_different}
\end{equation}
\hfill $\diamond$
\end{lemma}

\begin{proof}
In the birthday paradox, if there are $m$ people and $n$ equally possible birthdays, the probability that all the $m$ people have different birthdays is close to $1$ whenever $m \ll \sqrt{n}$. Here in our problem, at a certain box, the $M$ cache slots are regarded as ``people'' and the $|\mathcal{C}|$ contents are regarded as ``birthdays.'' Although the probability of picking one content is non-uniform, the probability of picking one content within a specific class is uniform. One can think of picking a content for a cache slot as a two-step process: With probability $\alpha_i \nu_i/\sum_{j}\alpha_{j}\nu_{j}$, a content in class $i$ is chosen. Then conditioned on class $i$, a specific content is chosen uniformly at random among all the $\alpha_i B$ contents in class $i$.

Contents from different classes are obviously different. When $M\ll \sqrt{\alpha_i B}$, even if all the $M$ cached contents are from class $i$, the probability that they are different is close to $1$. Thus, $M\ll \sqrt{\min_i \alpha_i B}$ is sufficient for (\ref{eq: almost_all_different}) to hold.
\end{proof}
\

To prove that under this particular placement, inefficiency in bandwidth utilization vanishes as $M\to\infty$, we shall in fact consider a slight modification of the ``request repacking'' strategy considered so far for determining which contents to accept:

\noindent\hrulefill\\
\textbf{Counter-Based Acceptance Rule}

\noindent\hrulefill\\
A parameter $L>0$ is fixed. Each box $b$ maintains at all times a counter $Z_b$ of associated requests. For any content $c$, the following procedure is used by the server whenever a request arrives: A random set of $L$ distinct boxes, each of which holds a replica of content $c$, is selected. An attempt is made to associate the newly arrived request with all $L$ boxes, but the request will be rejected if its acceptance would lead any of the corresponding box counters to exceed $LU$.

\noindent\hrulefill\\

\begin{remark}
Note that in this acceptance rule, associating a request to a set of $L$ boxes does not mean that the requested content will be downloaded from all these $L$ boxes. In fact, as before, the download stream will only come from one of the $L$ boxes, but here we do not specify which one is to be picked.

It is readily seen that the above rule defines a loss network. Moreover, it is a stricter acceptance rule than the previously considered one. Indeed, it can be verified that when all ongoing requests have an associated set of $L$ boxes, whose counters are no larger than $LU$, there exist nonnegative integers $Z_{cb}$
such that 
$\sum_{b: c\in \mathcal{J}_b} Z_{cb} =  L n_c,~\forall~c\in \mathcal{C}$ and $\sum_{c: c\in \mathcal{J}_b} Z_{cb} \le LU,~\forall~b \in \mathcal{B}$, then feasibility condition \eqref{eq: feasible_hall} 
holds a fortiori.
\hfill $\diamond$
\end{remark}
\

We introduce an additional assumption, needed for technical reasons.
\begin{assumption}
\label{assume: ignore_poor}
A content which is too poorly replicated is never served. Specifically, \textmathbf{a content must be replicated at least $M^{3/4}$ times to be eligible for service.}\hfill$\diamond$
\end{assumption}
\

Our main result in this context is the following theorem:
\begin{theorem}
\label{thm: suff_large_catalog}
Consider  fixed $M$, $\alpha_i$, $\nu_i$, and corresponding load $\rho<1$. Then for suitable choice of parameter $L$, with high probability (with respect to placement) as $B\to\infty$, the loss network with the above ``modified proportional-to-product placement'' and ``counter-based acceptance rule'' admits a content rejection probability $\phi(M)$ for some function $\phi(M)$ decreasing to zero as $M\to\infty$. \hfill $\diamond$
\end{theorem}
\

The interpretation of this theorem is as follows: The fraction of lost service opportunities, for an underloaded system ($\rho<1$), vanishes as $M$ increases. Thus, while Corollary \ref{col_1} showed that $M\to\infty$ is necessary for optimal performance, this theorem shows that it is also sufficient: there is no need for a minimal speed (e.g. $M\ge \log B$) to ensure that the loss rate becomes negligible. 


The proof is given in Appendix \ref{sec: large_catalog_proof}.

\section{Conclusion}
\label{sec: conclusion}
In peer-to-peer video-on-demand systems, the information of content popularity can be utilized to design optimal content placement strategies, which minimizes the fraction of rejected requests in the system, or equivalently, maximizes the utilization of peers' uplink bandwidth resources. We focused on P2P systems where the number of users is large. For the limited content catalogue size scenario, we proved the optimality of a proportional-to-product placement in the Distributed Server Network architecture, and proved optimality of ``Hot-Warm-Cold'' placement in the Pure P2P Network architecture. For the large content catalogue scenario, we also established that proportional-to-product placement leads to optimal performance in the Distributed Server Network. Many interesting questions remain. To name only two, more general popularity distributions (e.g. Zipf) for the large catalogue scenario could be investigated; the efficiency of adaptive cache update rules such as the one discussed in Section~\ref{sec: proportional}, or classical alternatives such as LRU, in conjunction with a loss network operation, also deserves more detailed analysis.

\bibliographystyle{abbrv}
\bibliography{P2P_VoD_Loss}

\appendix

\subsection{Proof of Theorem \ref{thm: suff_large_catalog}}
\label{sec: large_catalog_proof}
The proof has five sequential stages:
\begin{flushleft}
{\em 1) The chance for a content to be ``good''}
\end{flushleft}

Let $N_c$ denote the number of replicas of content $c$ of class $i$. Then,  $N_c$ admits a binomial distribution with parameters $(M\!B, \frac{\nu_i}{\rho BU})$. We call content $c$ a ``good'' content if $|N_c - \EE [N_c]| < M^{2/3}$, i.e.,
\begin{equation}
\left|N_c- \frac{\nu_i M}{\rho U} \right| < M^{2/3}.\label{eq: good_content_def}
\end{equation}
As $N_c = \sum_{i=1}^{M\!B} Z_i$, where $Z_i \sim Ber(p)$ ($p\triangleq \frac{\nu_i}{\rho BU}$) are i.i.d., according to the Chernoff bound,
\begin{equation}
\Pr\left(N_c \ge M^{2/3} + \frac{\nu_i M}{\rho U}\right) \le e^{-M\!B\cdot I(a)},	\label{eq: chernoff1}
\end{equation}
where $a\triangleq \left(M^{2/3} + \frac{\nu_i M}{\rho U}\right) / M\!B$ and $I(x) \triangleq \sup_{\theta} \{x \theta - \ln (\EE[e^{\theta Z_i}])\}$ is the Cram\'er transform of the Bernoulli random variable $Z_i$. Instead of directly deriving the RHS of inequality \eqref{eq: chernoff1}, which can be done but needs a lot of calculations (see Appendix \ref{sec: ber_chernoff}), we upper bound it by using a much simpler approach here: For the same deviation, a classical upper bound on the Chernoff bound of a binomial random variable is provided by the Chernoff bound of a Poisson random variable which has the same mean (see e.g. \cite{DraMas_epidemics10}). Therefore, the RHS of inequality \eqref{eq: chernoff1} can be upper bounded by 
$$
\exp\left(-\frac{\nu_i M}{\rho U} \cdot \hat{I}\left(1+\frac{\rho U}{\nu_i M^{1/3}}\right)\right),
$$
where $\hat{I}(x)$ is the Cram\'er transform of a unit mean Poisson random variable, i.e., $\hat{I}(x)=x\log x-x+1$. By Taylor's expansion of $\hat{I}(x)$ at $x=1$, the exponent in the last expression is equivalent to
\begin{eqnarray*}
&&-\frac{\nu_i M}{\rho U}\cdot \left(\frac{1}{2}\left(\frac{\rho U}{\nu_i M^{1/3}}\right)^2+o\left(M^{-2/3}\right)\right)\\
&=&-\frac{\rho U}{2\nu_i}M^{1/3} + o\left(M^{1/3}\right) = -\Theta\left(M^{1/3}\right). \label{eq: exp_order}
\end{eqnarray*}


On the other hand, when $M$ is large, $-M^{2/3} + \frac{\nu_i M}{\rho U} \ge 0$ holds, hence we have
\begin{eqnarray}
&&\Pr\left(N_c \le -M^{2/3} + \frac{\nu_i M}{\rho U}\right)\nonumber\\
&=& \Pr\left(\sum_{i=1}^{M\!B} \hat{Z}_i \ge M\!B\cdot \hat{a}\right)
\le e^{-M\!B \cdot \hat{I}(\hat{a})},\label{eq: chernoff2}
\end{eqnarray}
where $(-\hat{Z}_i) \sim Ber(p)$, $\hat{a}\triangleq M^{-1/3}/B - p \in [-1,0]$ when $B$ is large, and it is easy to check that $\hat{I}(\hat{a}) = I(-\hat{a})$. Similarly as above by upper bounding $e^{-MB\cdot I(-\hat{a})}$, we can find that the exponent of the upper bound is also $-\Theta\left(M^{1/3}\right)$. Therefore,
\begin{equation}
\Pr(\mbox{content $c$ is good}) \ge 1- 2 e^{-\Theta(M^{1/3})}.\label{eq: good_bound}
\end{equation}

\begin{flushleft}
{\em 2) The number of ``good contents'' in each class}
\end{flushleft}

Denoting by $X_i$ the number of good contents in class $i$, we want to use a corollary of Azuma-Hoeffding inequality (see e.g. Section 12.5.1 in \cite{MitUpf05_ProbComp} or Corollary 6.4 in \cite{DraMas_epidemics10}) to upper bound the chance of its deviation from its mean. This corollary applies to a function $f$ of independent variables $\xi_1,\ldots, \xi_{n}$, and states that if the function changes by an amount no more than some constant $c$ when only one component $\xi_i$ has its value changed, then for all $t>0$,
$$
\Pr(|f(\bm{\xi})-\EE[f(\bm{\xi})]|\ge t)\le 2 e^{-2t^2/(n c^2)}.
$$

Back to our problem, each independent variable $\xi_j$ correspond to the choice of a content to be placed in a particular memory slot at a particular box (we index a slot by $j$ for $1 \le j \le M\!B$), and $f(\bm{\xi})$ corresponds to the number of good contents in class $i$ based on the placement $\bm{\xi}$, i.e., $X_i= f(\bm{\xi})$. It is easy to see that in our case $c=1$, hence we have
\begin{equation*}
\Pr(|X_i-\EE [X_i]|\ge t)\le 2 e^{-2t^2/(M\!B)},~\forall t>0. 
\end{equation*}
Taking $t=(M\!B)^{2/3}$ in the above inequality further yields
$$\Pr\left(|X_i-\EE [X_i]| \ge (M\!B)^{2/3}\right)\le 2e^{-2(M\!B)^{1/3}}.$$ 
Thus,
we have
\begin{eqnarray}
&& \Pr\left(X_i\ge  \left(1-2 e^{-\Theta(M^{1/3})}\right) \cdot \alpha_i B -(MB)^{2/3}\right)\nonumber\\
&\stackrel{(a)}{\ge}&\Pr\left(X_i\ge  \EE[X_i]-(MB)^{2/3}\right)\nonumber\\
&\ge &  \Pr\left(|X_i - \EE[X_i]| <(MB)^{2/3}\right)\nonumber\\
&\ge & 1 - 2 e^{-2(M\!B)^{1/3}},
\label{eq: good_num_bound}
\end{eqnarray}
where (a) holds since
\begin{eqnarray*}
\EE[X_i] &=&\Pr(\mbox{content $c$ is good}) \cdot \alpha_i B \\
&\ge& \left(1- 2 e^{-\Theta(M^{1/3})}\right)\cdot \alpha_i B.
\end{eqnarray*}
Note that in order for the lower bound on $X_i$ shown in the above probability to be $\Theta(B)$, $M \sim o(B^{1/2})$ is a sufficient condition.

\begin{flushleft}
{\em 3) The chance for a box to be ``good''}
\end{flushleft}

We call a replica ``good'' if it is a replica of a good content, and use $C_i$ to denote the number of good replicas of class $i$. We also call a box ``good'' if the number of good replicas of class $i$ held by this box lies within
$$\frac{\alpha_i \nu_i M}{\rho U} \pm O(M^{2/3}).$$
As we did for ``good contents,'' we will also use the Chernoff bound to prove that a box is good with high probability.


Let ${\mathcal E}_i$ represent an event that the number $X_i$ of good contents within class $i$ satisfies
\begin{equation}
X_i\ge \left(1-2e^{-\Theta(M^{1/3})}\right)\alpha_i B - (M\!B)^{2/3}, \label{eq: good_num_LB}
\end{equation}
which has a probability of at least $1-2 e^{-\Omega((M\!B)^{1/3})}$,  according to inequality \eqref{eq: good_num_bound} when $M \sim o(B^{1/2})$. Conditional on $\mathcal{E}_i$, according to the lower bound in inequality \eqref{eq: good_content_def} (i.e., the definition of ``good contents'')  and inequality \eqref{eq: good_num_LB}, we have
\begin{eqnarray}
C_i &\ge& \left(\frac{\nu_i M}{\rho U}-M^{2/3}\right)
\bigg(\left(1-2 e^{-\Theta(M^{1/3})}\right)\alpha_i B \nonumber\\
&&-(M\!B)^{2/3}\bigg) \nonumber
\\
&=&M\!B\cdot \frac{\alpha_i \nu_i}{\rho U}\left(1-O(M^{-1/3}+M^{2/3} B^{-1/3})\right). \nonumber\\
\label{eq: Ci_LB}
\end{eqnarray}
On the other hand, from the upper bound in inequality \eqref{eq: good_content_def} and the fact $X_i \le \alpha_i B$, we obtain that
\begin{equation}
C_i\le M\!B\cdot\frac{\alpha_i \nu_i}{\rho U}\left(1+O(M^{-1/3})\right).\label{eq: Ci_UB}
\end{equation}
Conditional on $\mathcal{E}_i$, to constitute a box, sample without replacement from the determined content replicas. Denote the number of good replicas of class $i$ stored in a particular box (say, box $b$) by $\zeta_i$, which actually represents the number of good replicas in the $M$ samples sampled without replacement from all the $M\!B$ replicas, among which $C_i$ are good ones (conditional on $\mathcal{E}_i$). This means that, conditional on $\mathcal{E}_i$, $\zeta_i$ follows a hypergeometric distribution $H(M\!B, C_i, M)$. It can be found that (see e.g. Theorem 1 in \cite{klenke10StocOrder}) conditional on $\mathcal{E}_i$,
$H_i \le_{st} \zeta_i \le_{st} G_i.$ Here, ``$\le_{st}$'' represents stochastic ordering, and 
\begin{eqnarray*}
G_i &\sim& \mbox{Bin}\left(M,\frac{\alpha_i\nu_i}{\rho U}\left(1+O(M^{-1/3})\right)\right),\\
H_i &\sim& \mbox{Bin}\left(M,\frac{\alpha_i\nu_i}{\rho U}\left(1-O(M^{-1/3}+ M^{2/3} B^{-1/3})\right)\right),
\end{eqnarray*}
where the second parameters of the distributions of $G_i$ and $H_i$ are determined according to inequalities \eqref{eq: Ci_UB} and \eqref{eq: Ci_LB} respectively.

We will see why we need these two ``binomial bounds'' on $\zeta_i$. By definition, 
\begin{eqnarray}
&&\Pr(\mbox{box $b$ is not good})\nonumber\\ 
&=& \Pr\left(\bigcup_{i\in \mathcal{I}} \left\{\left|\zeta_i - \frac{\alpha_i \nu_i M}{\rho U}  \right|\ge O(M^{2/3})\right\}\right) \nonumber\\
&\le & \sum_{i \in \mathcal{I}}\Pr\left(\left|\zeta_i - \frac{\alpha_i \nu_i M}{\rho U} \right|\ge O(M^{2/3})\right), \label{eq: union_bound}
\end{eqnarray}
where for all $i \in \mathcal{I}$,
\begin{eqnarray}
&&\Pr\left(\left|\zeta_i - \frac{\alpha_i \nu_i M}{\rho U}  \right|\ge O(M^{2/3})\right)\nonumber\\
&=& \Pr\left(\left|\zeta_i - \frac{\alpha_i \nu_i M}{\rho U}  \right|\ge O(M^{2/3}),~\mathcal{E}_i\right) \nonumber\\
&&+ \Pr\left(\left|\zeta_i -\frac{\alpha_i \nu_i M}{\rho U}  \right|\ge O(M^{2/3}),~\mathcal{E}_i^c\right) \nonumber\\
&\le& \Pr\left(\left|\zeta_i - \left.\frac{\alpha_i \nu_i M}{\rho U}  \right|\ge O(M^{2/3})\right|\mathcal{E}_i\right) \cdot \Pr\left(\mathcal{E}_i\right)\nonumber\\
&&+ \Pr\left(\mathcal{E}_i^c\right). \label{eq: dev_bound1}
\end{eqnarray}
By definition of stochastic ordering, 
\begin{eqnarray*}
&&\Pr\left(\left|\zeta_i - \left.\frac{\alpha_i \nu_i M}{\rho U}  \right|\ge O(M^{2/3})\right|\mathcal{E}_i\right)\\
&\le & \Pr\left(G_i \ge \frac{\alpha_i \nu_i M}{\rho U} + O(M^{2/3})\right)\\
&& + \Pr\left(H_i \le \frac{\alpha_i \nu_i M}{\rho U} - O(M^{2/3})\right)\\
&\stackrel{(a)}{\le}& 2 e^{-\Theta(M^{1/3})},
\end{eqnarray*}
where (a) can be obtained using a similar Chernoff bounding approach as for $N_c$ in Stage 1 of this proof.
Thus, continuing from inequality \eqref{eq: dev_bound1}, we further have
\begin{eqnarray}
&&\Pr\left(\left|\zeta_i - \frac{\alpha_i \nu_i M}{\rho U}  \right|\ge O(M^{2/3})\right) \nonumber\\
&\le& 2 e^{-\Theta(M^{1/3})} \cdot \Pr(\mathcal{E}_i) + (1-\Pr\left(\mathcal{E}_i\right))\nonumber\\
&=&1-(1-2e^{-\Theta(M^{1/3})})\Pr\left(\mathcal{E}_i\right)\nonumber\\
&\le&1-(1-2 e^{-\Theta(M^{1/3})})(1-2 e^{-\Omega((M\!B)^{1/3})})\nonumber\\
&=&2 e^{-\Theta(M^{1/3})}- 2 e^{-\Omega((M\!B)^{1/3})}.\label{eq: dev_bound2}
\end{eqnarray}
Putting inequality \eqref{eq: dev_bound2} back to inequality \eqref{eq: union_bound} immediately results in
\begin{equation}
\Pr(\mbox{box $b$ is good}) \ge 1- 2|\mathcal{I}| e^{-\Theta(M^{1/3})}.\label{eq: good_box_bound}
\end{equation}

\begin{flushleft}
{\em 4) The number of ``good boxes''}
\end{flushleft}

We use a similar approach as in Stage 2 to bound the number of good boxes, say $Y$, which can be represented as a function $g(\bm{\xi})$ where $\bm{\xi} = (\xi_1,\xi_2,\cdots,\xi_{M\!B})$ is the same content placement vector defined in Stage 2. Still, $g(\bm{\xi})$ changes by an amount no more than $1$ when only one component $\xi_i$ has its value changed, then for all $t>0$, $\Pr(|Y-\EE[Y]|\ge t)\le 2 e^{-2t^2/(M\!B)}$, and taking $t=(M\!B)^{2/3}$ further yields
$$\Pr\left(|Y-\EE [Y]| \ge (M\!B)^{2/3}\right)\le 2e^{-2(M\!B)^{1/3}}.$$       
Similarly as we obtain inequality \eqref{eq: good_num_bound}, we finally come to
\begin{equation}
\Pr\left(Y \ge B\left(1-2|\mathcal{I}|e^{-\Theta(M^{1/3})}\right)\right) \ge 1-2e^{-2(M\!B)^{1/3}}. \label{eq: good_box_num_bound}
\end{equation}

\begin{flushleft}
{\em 5) The performance of a loss network}
\end{flushleft}

Finally, consider the performance of the loss network defined by the 
``Counter-Based Acceptance Rule.'' We introduce an auxiliary system to establish an upper bound on the rejection rate. In the auxiliary system, upon arrival of a request for content $c$, $L$ different requests are mapped to $L$ distinct boxes holding a replica of $c$, but here they are accepted or rejected individually rather than jointly. Letting $Z_b$ (respectively, $Z'_b$) denote the number of requests associated to box $b$ in the original (respectively, auxiliary) system, one readily sees that $Z_b\le Z'_b$ at all times and all boxes and for each box $b$, the process $Z'_b$ evolves as a one-dimensional loss network. We now want to upper bound the overall arrival rate of requests to a good box:\\

{\em (a) Non-good contents}

Assume that upon a request arrival, we indeed pick $L$ content replicas, rather than $L$ distinct boxes holding the requested content (as specified in the acceptance rule). This entails that, if two replicas of this content are present at one box, then this box can be picked twice. However, since a vanishing fraction of boxes will have more than one replicas of the same content when $M\ll  \sqrt{(\min_i \alpha_i)B}$ (as proved in Lemma \ref{lemma: suff_for_diff_replica}), we can strengthen the definition of a ``good'' box to ensure that, on top of the previous properties, a good box should hold $M$ distinct replicas. It is easy to see that the fraction of good boxes will still be of the same order as with the original weaker definition. 

With these modified definitions, consider one non-good content $c$ of class $i$ cached at a good box. Its unique replica will be picked with probability $L/N_c$ when the sampling of $L$ replicas among the $N_c$ existing ones is performed. Thus, since we ignore requests for all content $c$ with $N_c\le M^{3/4}$ (according to Assumption \ref{assume: ignore_poor}), the request rate will be at most $\nu_{i} L M^{-3/4}$.

Besides, there are at most $O(M^{2/3})$ non-good content replicas held by one good box. The reason is as follows: By definition, a good box holds at least 
\begin{equation}
\sum_{i\in \mathcal{I}} \left(\frac{\alpha_i \nu_i M}{\rho U}-O(M^{2/3})\right) = M -O(M^{2/3})  \label{eq: good_replica_in_good_box}
\end{equation}
good content replicas among all classes, so the remaining slots, being occupied by non-good content replicas, are at most $O(M^{2/3})$. Therefore, the overall arrival rate of requests for non-good contents to a good box is upper bounded by
\begin{equation}
\overline{\nu}_{\mbox{non-good}} = O(M^{2/3}\cdot L M^{-3/4}) = O(LM^{-1/12}). \label{eq: non-good_rate_UB}
\end{equation}

{\em (b) Good contents}

The rate generated by a good content $c$ of class $i$ is $\nu_i L/N_c$. Now, by definition of a good content, one has:
$$
N_c\ge \frac{\nu_i M}{\rho U}(1-O(M^{-1/3})).
$$
This entails that the rate of requests for this content is upper bounded by
$$
\frac{\rho LU}{M}(1+O(M^{-1/3})).
$$
By definition of a ``good box,'' there are at most $\alpha_i\nu_i M/\rho U +O(M^{2/3})$ good content replicas of class $i$ cached in this good box. Therefore, the overall arrival rate of requests for good contents to a good box is upper bounded by
\begin{eqnarray}
\overline{\nu}_{\mbox{good}} &=& \sum_{i\in \mathcal{I}} \left( \frac{\rho L U}{M}(1+O(M^{-1/3}))\right) \nonumber\\
&&\times \left(\frac{\alpha_i\nu_i M}{\rho U} +O(M^{2/3})\right) \nonumber\\
&=& (\rho LU)(1+O(M^{-1/3})).\label{eq: good_rate_UB}
\end{eqnarray}

To conclude, for any good box $b$, the process $Z'_b$ evolves as a one-dimensional loss network with arrival rate no larger than
$$
\overline{\nu} = \overline{\nu}_{\mbox{non-good}} + \overline{\nu}_{\mbox{good}} = \rho LU+ O(LM^{-1/12}),
$$
by combining the two results in \eqref{eq: non-good_rate_UB} and \eqref{eq: good_rate_UB}.\\

Next, we are going to upper bound the loss probability of $Z'_b$. Since $\overline{\nu}$ is an upper bound on the arrival rate, the probability that $Z'_b=LU$ is upper bounded by $E(\rho LU + O(LM^{-1/12}),LU)$. One can actually further upper bound this Erlang function by $e^{-\Theta(L)}$. To see this, 
let us first rewrite the loss probability (Erlang function) of a general 1-D loss network, say $E(\lambda, C)$, as a certain conditional probability of $S\sim \mbox{Poi}(\lambda)$, i.e.,
$$
E(\lambda,C) = \Pr(S=C|S\le C) = \frac{\Pr(S=C)}{\Pr(S\le C)}.  
$$
Using the Chernoff bound, we have $\Pr(S\ge C) \le e^{-\lambda I(C/\lambda)}$,  where $I(x) = x \log x - x + 1$, hence
$$
E(\lambda,C) \le \frac{\Pr(S\ge C)}{1-\Pr(S \ge C)} \le \frac{e^{-\lambda I(C/\lambda)}}{1-e^{-\lambda I(C/\lambda)}}.
$$
Back to the Erlang function in our problem, $I(C/\lambda) = I((\rho+O(M^{-1/12}))^{-1})$, hence, 
\begin{equation}
\Pr(Z'_b = LU) \le E(\rho LU + O(LM^{-1/12}),LU) \le e^{-\Theta(L)},\label{eq: loss_net_small_loss}
\end{equation}
where the second inequality holds under the assumption that $\rho < 1$ (otherwise, the exponent will become $0$ or $+\Theta(L)$). \\

The number of good replicas in good boxes is, due to inequality \eqref{eq: good_box_num_bound} and equation \eqref{eq: good_replica_in_good_box}, at least $M\!B(1-O(M^{-1/3}))$, with a high probability (at least $1-2e^{-2(M\!B)^{1/3}}$). On the other hand, the total number of replicas of good contents is at most $M\!B$, which is the total number of replicas (or available cache slots).

Now pick some small $\epsilon \in (0,1/3)$ and let $\tilde{X}$ denote the number of good contents which have at least $M^{2/3+\epsilon}$ replicas outside good boxes. Then necessarily, with a probability of at least $1-2e^{-2(M\!B)^{1/3}}$,
\begin{eqnarray*}
&&\tilde{X}\! M^{2/3+\epsilon}\le M\!B-M\!B(1-O(M^{-1/3}))
=O(B\!M^{2/3}),
\end{eqnarray*}
i.e., $\tilde{X}\le O(B\!M^{-\epsilon})$. According to inequality \eqref{eq: good_num_bound}, the total number of good contents is $\Theta(B)$ (specifically, very close to $|\mathcal{C}| = \alpha B$) with a probability of at least $1-2|\mathcal{I}|e^{-2(M\!B)^{1/3}}$, hence we can conclude that, with high probability, for a fraction of at least $1-O(M^{-\epsilon})$ of good contents, each of them has 
at least a fraction $1-O(M^{-1/3+\epsilon})$ of its replicas stored in good boxes (since a good content has $\frac{\nu_i}{\rho U}M \pm O(M^{2/3})$ replicas in total by definition).
We further use $\tilde{\mathcal{C}}$ to represent the set of such contents.

Recall that $A_c$ was defined in Subsection \ref{sec: loss_network} as the steady-state probability of accepting a request for content $c$ in the original system. For all $c \in \tilde{\mathcal{C}}$, 
\begin{eqnarray}
A_c &\ge & \Pr(\mbox{all the $L$ sampled replicas are in good boxes}) \nonumber\\
&& \times \Pr(Z_b < LU,~\forall b~s.t.~\mbox{box $b$ is sampled})\nonumber\\
&\stackrel{(a)}{\ge}&\left(1-O(M^{-1/3+\epsilon})\right)^L\nonumber\\
&& \times \Pr(Z'_b < LU,~\forall b~s.t.~\mbox{box $b$ is sampled}).\nonumber\\
&\stackrel{(b)}{\ge} & \left(1-O(M^{-1/3+\epsilon})\right)^L \cdot \left(1-L e^{- \Theta(L)}\right). \nonumber
\\\label{eq: accept_LB}
\end{eqnarray}
Here, (b) is obtained according to inequality \eqref{eq: loss_net_small_loss}. The argument why (a) holds is as follows: 
We have $N_c\approx \nu_i M/(\rho U)$ replicas (assuming that content $c$ is of class $i$), among which $N'_c=N_c(1-O(M^{-1/3+\epsilon}))$ are in good boxes. Then, the probability that $L$ samples fall in the good boxes can be written explicitly as 
$$\frac{N'_c(N'_c-1)\cdots (N'_c-L+1)}{N_c (N_c-1)\cdots (N_c-L+1)},$$ 
which can be approximated as the first part on the RHS we write above, under the assumption that $L \ll M$. The second part is due to the fact that $Z'_b \le Z_b$ for all box $b$.


It should be recalled that within this stage of proof, finally coming to inequality \eqref{eq: accept_LB} actually needs everything to be conditional on the following events:
\begin{itemize}
\item The number of good boxes is $\Theta(B)$;
\item The number of good contents is $\Theta(B)$;
\item A box caches $M$ distinct replicas, 
\end{itemize}  
and as $B, M\rightarrow \infty$ and $M\ll  \sqrt{(\min_i \alpha_i)B}$, all of them have high probabilities. Additionally, $\tilde{\mathcal{C}}\stackrel{p}{\rightarrow}\mathcal{C}$ as $B,M\rightarrow \infty$. Therefore, further letting $L\rightarrow \infty$ but keeping $L \ll M^{1/3-\epsilon}$, we will find that the RHS of inequality \eqref{eq: accept_LB} is approximated as
$$
1-O(LM^{-1/3+\epsilon})-L e^{- \Theta(L)} \approx 1,
$$
and then conclude that the requests for almost all the contents will have near-zero loss.

\subsection{Proof of Equivalence between Feasibility Conditions~(\ref{eq: feasible_ori}) and~(\ref{eq: feasible_hall})}
\label{sec: proof_hall}
\subsubsection{Sufficiency of Condition~(\ref{eq: feasible_hall})}
We use Hall's theorem to prove the sufficiency.\\\\
\textbf{[Hall's theorem] }Suppose $\mathcal{J} = \{J_1,J_2,\cdots\}$ is a collection of sets (not necessarily countable). A SDR (``System of Distinct Representatives'') for $\mathcal{J}$ is defined as $X = \{x_1, x_2, \cdots\}$, where $x_i\in J_i$. Then, there exists a SDR (not necessarily unique) iff. $\mathcal{J}$ meets the following condition:
\begin{equation}
\forall~\mathcal{T}\subseteq \mathcal{J},~|\mathcal{T}|\le |\bigcup_{A\in\mathcal{T}}A|. \label{eq: marriage_cond}
\end{equation}
\hfill $\diamond$

In our P2P VoD system, denote the content set as $\mathcal{C} = \{c_1, c_2, \cdots,c_N\}$. Given the ongoing download services of each content $\{n_i\}_{i=1}^N$, we get a ``distinguishable content set''
\begin{eqnarray*}
\bar{\mathcal{C}} &=& \{c_1^{(1)}, c_1^{(2)}, \cdots, c_1^{(n_1)}; c_2^{(1)}, c_2^{(2)}, \cdots, c_1^{(n_2)}; \cdots;\\
&& c_N^{(1)}, c_N^{(2)}, \cdots, c_N^{(n_N)}\},
\end{eqnarray*}
where $c_i^{(k)}$ represents the $k$-th download service of content $i$ for $1\le k \le n_i$, and has its ``potential connection set''
$$J_i^{(k)} = \{l_b^{(j)}:~1\le j\le U,~c_i\in b,~b\in\mathcal{B}\},$$
i.e., the set of all the connections of those boxes which have content $c_i$. A collection of the ``potential connection sets'' for all $\{c_i^{(k)}\}$ is then
$$\mathcal{J} = \{J_1^{(1)}, J_1^{(2)}, \cdots, J_1^{(n_1)}; \cdots; J_N^{(1)}, J_N^{(2)}, \cdots, J_N^{(n_N)}\},$$
and a SDR for $\mathcal{S}$ is
$$X = \{x_1^{(1)}, x_1^{(2)}, \cdots, x_1^{(n_1)}; \cdots; x_N^{(1)}, x_N^{(2)}, \cdots, x_N^{(n_N)}\},$$
s.t. $x_i^{(k)}\in J_i^{(k)},$ which means each $c_i^{(k)}$ is affiliated with a distinct connection (i.e., a feasible solution in our model).

Now we want to prove the existence of such a SDR, i.e., to prove equation~(\ref{eq: marriage_cond}). For $\forall~\mathcal{T}\subseteq \mathcal{J}$, there is a one-to-one mapping between $\mathcal{T}$ and a $\bar{\mathcal{S}}\subseteq \bar{\mathcal{C}}$. Further, this $\bar{\mathcal{S}}$ can be mapped to a $\mathcal{S}\subseteq \mathcal{C}$ where
$$\mathcal{S} = \{c_i:~\exists 1\le k\le n_i,~s.t.~c_i^{(k)} \in \bar{\mathcal{S}}\},$$
i.e., $\mathcal{S}$ is the set of all contents considered in $\bar{\mathcal{S}}$ without considering multiple services of each content. Then, $\forall~\mathcal{T}\subseteq \mathcal{J}$,
\begin{eqnarray*}
\mbox{RHS} &=& |\bigcup_{J^{(k)}_i\in \mathcal{T}}J_i^{(k)}| = \sum_{b:\exists c_i\in \mathcal{S}~s.t.~c_i\in b} U\\
&=& U \left|\{b\in{\mathcal B}:\;  \mathcal{S}\cap  \mathcal{J}_b\ne \emptyset\}\right|
\end{eqnarray*}
and
$$\mbox{LHS} = |\mathcal{T}| = |\bar{\mathcal{S}}| \le \sum_{c_i\in\mathcal{S}}n_i.$$

Therefore, if
$$\forall S\subseteq \mathcal{C}, \sum_{c_i\in\mathcal{S}}n_i \le = U \left|\{b\in{\mathcal B}:\;  \mathcal{S}\cap  \mathcal{J}_b\ne \emptyset\}\right|$$
holds, then equation~(\ref{eq: marriage_cond}) holds. The sufficiency is proved.\\

\subsubsection{Necessity of Condition~(\ref{eq: feasible_hall})}
For any $\mathcal{S}\subseteq\mathcal{C}$,
\begin{eqnarray*}
\sum_{c\in\mathcal{S}}n_c &=& \sum_{c\in\mathcal{S}}\sum_{b:c\in \mathcal{J}_b}Z_{cb}=\sum_{{b:~\exists c \in \mathcal{S}}\atop {s.t.~c\in \mathcal{J}_b}} \sum_{c\in\mathcal{S}\cap\mathcal{J}_b} Z_{cb}\\
&\stackrel{(a)}{\le}& \sum_{b:~\exists c \in \mathcal{S}~s.t.~c\in \mathcal{J}_b} U = U \left|\{b\in{\mathcal B}:\;  \mathcal{S}\cap  \mathcal{J}_b\ne \emptyset\}\right|,
\end{eqnarray*}
where the inequality (a) is due to the second constraint in condition~(\ref{eq: feasible_ori}). Hence, the necessity is proved.

\subsection{Approximation to Proportional-to-Product Placement Using Bernoulli Sampling}
\label{sec: analyse_push}
An alternative sampling strategy to get the proportional-to-product placement is as follows:


\noindent\hrulefill\\
To push contents to box $b$ $(1\le b \le B)$, the server will
\begin{enumerate*}
\item Generate $C$ independent Bernoulli random variables $X_c \sim \mbox{Ber}(p_c)$ for all $c\in\mathcal{C}$, where $p_c = \beta \hat{\nu}_c / (1 + \beta \hat{\nu}_c)$, $\hat{\nu}_c$ is the normalized version of $\nu_c$, and $\beta$ is a customized constant parameter. \label{step: gen_ber_rv}
\item If $\sum_{c\in \mathcal{C}} X_c = M$ (which means a valid cluster of size $M$ is generated), push content $c$ to box $b$ if $X_c = 1$; Otherwise, go back to Step~\ref{step: gen_ber_rv}.
\end{enumerate*}
\noindent\hrulefill \vspace{0.1in}

We now analyze why this scheme works: after generating a valid size-$M$ subset, the probability that this subset is a certain subset $\mathcal{G}_j$ equals
\begin{eqnarray*}
&&\Pr(X_c = 1,~\forall c\in \mathcal{G}_j;~X_c = 0,~\forall c\not\in \mathcal{G}_j | \sum_{c\in \mathcal{C}} X_c = M)\\
&&= \frac{\prod_{c\in \mathcal{G}_j}p_c \cdot \prod_{c\not\in \mathcal{G}_j}(1-p_c)}{\Pr(\sum_{c\in \mathcal{C}} X_c = M)}\\
&&= \prod_{c\in \mathcal{G}_j} \frac{p_c}{1-p_c} \cdot \left(\frac{\prod_{c\in\mathcal{C}}p_c}{\Pr(\sum_{c\in \mathcal{C}} X_c = M)}\right)\\
&&= \prod_{c\in \mathcal{G}_j} \hat{\nu}_c / Z,
\end{eqnarray*}
where $Z = \Pr(\sum_{c\in \mathcal{C}} X_c = M) / (\beta^M \prod_{c\in\mathcal{C}}p_c)$, which actually equals the normalizing factor for $\prod_{c\in \mathcal{G}_j} \hat{\nu}_c$.

We then consider the computational complexity of this approximation algorithm. Assuming that $\{\hat{\nu}_c\}$ is sorted in the descending order, we have \begin{eqnarray*}
\Pr(\sum_{c\in \mathcal{C}} X_c = M) &\ge& \prod_{c=1}^M p_c \cdot \prod_{c=M+1}^{C} (1-p_c) \\
&=& \frac{\prod_{c=1}^M \beta \hat{\nu}_c}{\prod_{c=1}^{C} (1+\beta \hat{\nu}_c)} \triangleq P^*.
\end{eqnarray*}
So the computational complexity is upper bounded by $O(BC/P^*)$. Note that the constant parameter $\beta$ can be adjusted to get a higher $\Pr(\sum_{c\in \mathcal{C}} X_c = M)$ in order to reduce computational complexity. To achieve this, we can just choose a $\beta$ which maximizes its lower bound $P^*$, so
\begin{equation}
\frac{\partial \log P^*}{\partial \beta} = \frac{M}{\beta} - \sum_{c=1}^{C} \frac{\hat{\nu}_c}{1+\beta \hat{\nu}_c} = 0. \label{eq: optimal_beta}
\end{equation}
The server can use any numerical methods (e.g., Newton's method) to seek a root of equation~(\ref{eq: optimal_beta}). In fact, this lower bound $P^*$ on $\Pr(\sum_{c\in \mathcal{C}} X_c = M)$ is not tight, since it is just the largest item in the sum expression. When the popularity is close to uniformness (e.g., in a zipf-like distribution, $\alpha$ is small), this largest item is no longer dominant, so the lower bound $P^*$ is quite untight, which means we actually overestimate the computation complexity by only evaluating its upper bound. However, this will not affect the real gain we obtain after choosing the optimal $\beta$ according to equation~(\ref{eq: optimal_beta}).

Recall that we also proposed a simple sampling strategy in Section~\ref{sec: proportional}. It is easy to see that when some contents are much more popular than the others (e.g., zipf-like $\alpha$ is large), the probability that duplicates appear in one size-$M$ sample is high, hence largely increases the number of resampling. Thus, it would be faster if we choose the Bernoulli sampling. However, when the popularity is quite uniform, the simple sampling works very well. An extreme case is that under the uniform popularity distribution,
$$\Pr\{\mbox{a valid size-$M$ subset}\} = \frac{{C\choose M} \cdot M!}{C^M} = \prod_{i=1}^{M-1} \left(1-\frac{i}{C}\right),$$
which shows that when $C$ is large, you can get a valid sample almost every time.

\subsection{Detailed Implementation in the Simulations}
\label{sec: implement}

\subsubsection{A Heuristic Repacking Algorithm}
\label{sec: repack}


%
%
%

We first describe the concept of ``repacking.'' When the cache size $M=1$, all the bandwidth resources at a certain box belongs to the content the box caches. When $M\ge 2$, however, this is not the case: all the contents cached in one box are actually competitors for the bandwidth resources at that box. Let's consider a simple example in which $B = 2$, $M=2$ and $U=1$: Box $1$ which caches content $1$ and $2$ is serving a download of content $2$, while box $2$ which caches content $2$ and $3$ is idle. When a request for content $1$ comes, the only potential candidate to serve it is box $1$, but since the only connection is already occupied by a download of content $2$, the request for content $1$ has to be rejected. However, if this ongoing download can be ``forwarded'' to the idle box $2$, the new request can be satisfied without breaking the old one. We call this type of forwarding ``repacking.''

In the the feasibility condition \eqref{eq: feasible_ori} and its equivalent form \eqref{eq: feasible_hall}, we actually allow perfect repacking to identify a feasible $\{n_c\}$. In a real system, perfect repacking needs to enumerate all the possible serving patterns and choose the best one based on some criterion, which is usually computationally infeasible. 
We then propose a heuristic repacking algorithm which is not so complex but can achieve similar functionality and improve performances, although imperfect. 

Several variables need to be defined before we describe the algorithm:

\begin{itemize}
\item $n_c$: the system-wide ongoing downloads of content $c$, which does not count the downloads from the server.
\item $\mathcal{B}_c^k$: The set of boxes which have content $c$ (``potential candidate boxes'') and $k$ free connections, for $0\le k \le U$.
\item $D_c$: number of boxes which has content $c$. $D_c = \sum_{k=0}^{U} |\mathcal{B}_c^k|$.
\item $\mathbf{u}_b$: a $U$-dimensional vector, of which the $i$-th component represents the content box $b$ is using its $i$-th connection to upload (a value $0$ represents a free connection).
\item $c_o$: the ``orphan content'' which is affiliated with a new request or an ongoing download but has not been assigned with any box.
\item $\mathcal{C}_o$: the set of contents which has once been chosen as orphan contents.
\item $t_R$: the number of repacking already done.
\end{itemize}

Note that when choosing a box to serve a request, load balancing is already considered, which to some extent reduces the chance of necessary repacking in later operations. However, repacking is still needed for an incoming request for content $c$ as soon as $\cup_{k>0}\mathcal{B}_c^k = \emptyset$. 

\noindent\hrulefill\\
\textbf{Repacking Algorithm}

\noindent\hrulefill\\
After getting a request for content $c$ while $\cup_{k>0}\mathcal{B}_c^k = \emptyset$, the server
\begin{enumerate}
\item Initialize $c_o := c$, $\mathcal{C}_o := \{c\}$, and $t_R := 0$.
\item Let $\bar{\mathcal{C}} = \{c': n_{c'} / D_{c'} > n_{c_o} / D_{c_o}~\mbox{and}~c'\not\in\mathcal{C}_o\}$, i.e., a set of contents which haven't become orphans during this repacking process and of which the utilization factor (may be larger than $1$) is larger than that of the current orphan content $c_o$. If $\bar{\mathcal{C}}_o = \emptyset$, regard $c_o$ as a loss and TERMINATE. \label{step: repack_compare}
\item Choose $c^* = \arg\max_{c'\in \bar{\mathcal{C}}}\{n_{c'} / D_{c'}\}$. Uniformly pick one (box, connection) pair from
$$\{(b,i):~b\in \mathcal{B}_c^0,~c^*\mbox{ is the $i$-th component of }\mathbf{u}_b\}.$$
\item Use the chosen box $b$ and its $i$-th connection to continue uploading the remaining part of content $c_o$. At the same time, $c^*$ which was served using that connection becomes a new orphan, i.e., $c_o := c^*$. Update $\mathbf{u}_b$ and $\{n_c\}$. Set $t_R:=t_R+1$.
\item If $\cup_{k>0}\mathcal{B}_{c_o}^k \not=\emptyset$, i.e., there exists a free connection to serve the new $c_o$, then use the load-balancing-based box selection rule to select a box to continue uploading the remaining part of $c_o$. The repacking process is perfect (no remaining orphan) and TERMINATE. Otherwise, \label{step: repack_rescue}
    \begin{itemize}
    \item If $t_R = t_R^{max}$, a customized algorithm parameter ($0 \le t_R^{max}\le C$), regard $c_o$ as a loss and TERMINATE.
    \item Otherwise, set $\mathcal{C}_o:=\mathcal{C}_o + \{c_o\}$, and go to Step~\ref{step: repack_compare}.
    \end{itemize}
\end{enumerate}
\noindent\hrulefill \vspace{0.1in}

\subsubsection{A Practical Issue in Cache Update}
When a box $b$ is chosen for cache update (and it does not hold the content $c$ corresponding to the request), it might still be uploading  content $c'$ which is to be replaced. This fact is not captured by the Markov chain model. In practice, those ongoing services must be terminated. Since we have introduced the repacking scheme, they become ``orphans'' ready for repacking. We implement the procedure as follows: 
\begin{enumerate}
    \item Rank these orphans by their remaining service time in the ascending order, i.e., the original download which is sooner to be completed is given higher priority.
    \item Do repacking one by one until one orphan fails to be repacked. Note that here the repacking algorithm starts from Step~\ref{step: repack_rescue}, since there may already be some boxes with both content $c$ and free connections.
\end{enumerate}


\subsection{Proof of Theorem~\ref{thm: wf_with_reservation}}
\label{sec: proof_wf_reserve}
The Lagrangian of OPT 2 is
\begin{eqnarray*}
&&L(\mathbf{\tilde{m}},\bm{\lambda}, \mathbf{x};\mathbf{u},\mathbf{v},\mathbf{y},\mathbf{z},\mathbf{w},\eta,\gamma)\\
&=& \sum_{c\in \mathcal{C}}\Big[\rho_c\tilde{m}_c+x_c - u_c(\tilde{m}_c-1)-v_c(\lambda_c-\tilde{m}_c)\\
&&- y_c(x_c-\lambda_c)-z_c(x_c-\rho_c+\rho_c\tilde{m}_c) + w_c\lambda_c\Big]\\
&&-\eta\left(\sum_{c\in \mathcal{C}}\lambda_c-1\right)-\gamma\left(\sum_{c\in \mathcal{C}}\tilde{m}_c-M\right).
\end{eqnarray*}
The KKT condition includes the feasible set defined in OPT 2 and the following:
\begin{eqnarray*}
\frac{\partial L}{\partial x_c} = 1- y_c - z_c &=& 0,~\forall c;\\
\frac{\partial L}{\partial \tilde{m}_c} = \rho_c- u_c + v_c - \rho_c z_c -\gamma &=& 0,~\forall c;\\
\frac{\partial L}{\partial \lambda_c} = - v_c + y_c - \eta + w_c &=& 0,~\forall c;\\
u_c(\tilde{m}_c-1)&=&0,~u_c\ge 0,~\forall c;\\
v_c(\lambda_c-\tilde{m}_c) &=&0,~v_c\ge 0,~\forall c;\\
y_c(x_c-\lambda_c) &=& 0,~y_c\ge 0,~\forall c;\\
z_c(x_c-\rho_c+\rho_c\tilde{m}_c) &=& 0,~z_c\ge 0,~\forall c;\\
w_c\lambda_c &=&0,~w_c \ge 0,~\forall c.
\end{eqnarray*}
We then put the solution stated in the theorem into KKT condition to check whether the condition is satisfied. The analysis is as follows:
\begin{itemize}
\item For $1\le c \le M-1$, since $\tilde{m}_c=1$ and $\lambda_c=x_c=0$, we obtain that $v_c = 0$, $y_c + z_c = 1$, $\rho_c(1-z_c) = u_c + \gamma$, and $y_c = \eta-w_c$. Letting $w_c = 0$, we further have:
    $u_c = \rho_c \eta - \gamma,~y_c = \eta,~z_c = 1-\eta.$ To keep $u_c,y_c,z_c\ge 0$, we must have $\eta \in [0,1]$ and $\gamma \le \rho_c\eta$, for $1\le c \le M-1.$ Thus, since $\{\rho_c\}$ are also ranked in the descending order, we have
    \begin{equation}
    \gamma \le \rho_{_{M-1}}\eta.\label{eq: opt2_sol_ineq1}
    \end{equation}
\item For $M \le c \le c^*$, since $\tilde{m}_c = \lambda_c = x_c = \rho_c/(1+\rho_c)$, we obtain that $u_c = w_c= 0,~y_c+z_c = 1,~\rho_c(1-z_c)=\gamma-v_c,~y_c = \eta+v_c$. We further have:
    $$v_c = \frac{\gamma-\rho_c\eta}{\rho_c+1},~y_c = \frac{\eta+\gamma}{\rho_c+1},~z_c=1-\frac{\eta+\gamma}{\rho_c+1}.$$
    To keep $v_c,y_c,z_c\ge 0$, we must have $\rho_c\eta \le \gamma \le \rho_c+1-\eta,~\mbox{for}~M\le c \le c^*.$ Thus,
    \begin{equation}
    \rho_{_{M}}\eta \le \gamma \le \rho_{c^*}+1-\eta. \label{eq: opt2_sol_ineq2}
    \end{equation}
\item For $c = c^*+1$, when $m_c =0$, it degenerates to the next case. When $\tilde{m}_c >0$, since $\tilde{m}_c=\lambda_c = x_c < \rho_c(1-\tilde{m}_c)$, we obtain that $u_c = w_c = z_c =0,~y_c=1,~\rho_c+v_c = \gamma,~\eta+v_c = 1.$ We further have
    \begin{equation}
    \gamma = \rho_{c^*+1}+1-\eta. \label{eq: opt2_sol_eq3}
    \end{equation}
\item For $c^*+2\le c \le C$, since $\tilde{m}_c =\lambda_c = x_c = 0$, we obtain that $u_c = z_c =0,~y_c = 1,~v_c = \gamma - \rho_c,~w_c = \eta+v_c-1 = \eta+\gamma-\rho_c-1$. To keep $v_c,w_c\ge 0,$ and due to the fact that $\eta\in [0,1]$, we must have $\gamma \ge \rho_c,~\mbox{for}~c^*+2 \le c \le C.$ Thus,
    \begin{equation}
    \gamma \ge \rho_{c^*+2}.\label{eq: opt2_sol_ineq4}
    \end{equation}
\end{itemize}
For inequalities~(\ref{eq: opt2_sol_ineq1}),~(\ref{eq: opt2_sol_ineq2}), (\ref{eq: opt2_sol_ineq4}) and equation~(\ref{eq: opt2_sol_eq3}) to hold simultaneously, we can choose a $\eta$ which satisfies
$$
\frac{\rho_{c^*+1}+1}{\rho_{_{M-1}}+1} \le \eta \le \frac{\rho_{c^*+1}+1}{\rho_{_{M}}+1},
$$
which also satisfies $\eta \in [0,1]$. Therefore, the theorem is proved.

It should be mentioned that when $\sum_{c=M}^{c^*}\tilde{m}_c = 1$, i.e., $\tilde{m}_{c^*+1} = 0$, the case ``$c = c^*+1$'' can be combined with the next case ``$c^*+2\le c \le C$'', hence equation~(\ref{eq: opt2_sol_eq3}) does not exist while inequality~(\ref{eq: opt2_sol_ineq4}) is changed to $\gamma \ge \rho_{c^*+1}.$ Then, we can just choose a $\eta$ which satisfies
$$
0\le \eta \le \frac{\rho_{c^*+1}+1}{\rho_{_{M}}+1}.
$$

\subsection{Storage of Segments and Parallel Substreaming}
\label{sec: parallel}
We have mentioned before that compared to the ``storage of complete contents and downloads by single streaming'' setting, a more widely used mechanism in practice is that each box stores one specific segment of a video content and a download (streaming) comprises parallel substreaming from different boxes. To model this mechanism, we have the following simplifying assumptions: Each content is divided into $K$ segments with equal length which are independently stored. Each box can store up to $M$ segments (actually it does not matter if we keep the original storage space of each box, i.e., $M$ complete contents, which now can hold $M\!K$ segments, since the storage space is a customized parameter) and these $M$ segments do not necessarily belong to $M$ distinct contents. The bandwidth of each box is kept as $U$, so now each box can accommodate $U\!K$ parallel substreaming, each with download rate $1/K$ (the average service duration is still kept as $1$ because each segment is $1/K$ of the original content length). The definition of ``traffic load'' $\rho$ is then the same as in equation~(\ref{eq: rho}). A request for a content will be divided into sub-requests submitted to the boxes holding those corresponding segments of this content, generating $K$ parallel substreaming flows in total (one box can serve more than one substreaming service for this request if it caches more than one distinct segments of this content).

Let $\theta$ represent a segment and $\theta\in c$ indicate that $\theta$ is a segment of content $c$. Recall that we use $n_c$ to denote the number of concurrent downloads (now called ``streams'') of content $c$ in the network. We further use $n_{\theta}$ to denote the number of substreams corresponding to segment $\theta$.

Now the original feasibility constraint~(\ref{eq: feasible_ori}) becomes
\begin{eqnarray}
\sum_{b:~\theta \in \mathcal{J}_b} z_{\theta b} & = & n_{\theta},~\forall~\theta\in \Theta; \nonumber\\
\sum_{\theta:~\theta\in \mathcal{J}_b} z_{\theta b} & \le & U\!K,~\forall~b \in \mathcal{B}, \label{eq: feasible_ori_parallel}
\end{eqnarray}
where $\Theta$ represents the whole set of segments and $z_{\theta b}$ denotes the the number of concurrent substreams downloading segment $\theta$ from box $b$. It is easy to see that the equivalent version which can be proved by Hall's theorem becomes:
\begin{equation}
\forall~\mathcal{S} \subseteq \Theta,~\sum_{\theta \in \mathcal{S}} n_{\theta}
\le KU \left|\{b\in{\mathcal B}:\;  \mathcal{S}\cap  \mathcal{J}_b\ne \emptyset\}\right|, \label{eq: feasible_hall_parallel}
\end{equation}
where with a little abuse of notation, $\mathcal{S}$ is used to denote a subset of $\Theta$, instead of $\mathcal{C}$ as before.

Since we have assumed that video duration and video streaming rate are all the same, one naturally has $n_{\theta} = n_c$ for all $\theta\in c$. If we let randomness exist in the service duration, then within one stream, some substreams may complete earlier than the others. Therefore, the above equality needs to be added as a constraint (and used to come up with the following result), i.e., the bandwidth for the $K$ substreams should be reserved until the whole streaming is completed.

Then, in the proof of the optimality of ``proportional-to-product'' placement for DSN, every expression keeps the same, except that the feasibility constraint~(\ref{eq: primal_constraint}) is changed to
\begin{equation}
\forall~\mathcal{S} \subseteq \Theta,~\sum_{\theta \in \Theta} \sum_{c:\theta\in c} x^{_{(B)}}_{c} \le \sum_{j: j \cap \mathcal{S} \not= \emptyset} m_j BU\!K, \label{eq: primal_constraint_parallel}
\end{equation}
and the ``proportional-to-product'' placement $\{m_j\}$ is now with respect to each segment, i.e., $m_j = \prod_{\theta\in j}\hat{\nu}_\theta /Z$ for all $j\subseteq \Theta$ s.t. $|j|=M$, where $Z$ is the normalizing constant and $\hat{\nu}_{\theta} = \hat{\nu}_{c}$ if $\theta \in c$. With an observation that $\sum_{\theta\in\Theta} \hat{\nu}_{\theta} = K \sum_{c\in\mathcal{C}} \hat{\nu}_{c} = K$, we can still come to an inequality same with inequality~(\ref{eq: feasible_prod}), except that $c$ and $\mathcal{C}$ are replaced by $\theta$ and $\Theta$ respectively. All the succeeding steps are exactly the same in the proof of optimality.

\subsection{Another Approach to Bound the Chance of ``Good Contents'' in Proving Theorem \ref{thm: suff_large_catalog}}
\label{sec: ber_chernoff}
At the first stage of proving Theorem \ref{thm: suff_large_catalog}, we mentioned that we can also directly derive the Chernoff bound on the RHS of inequality \eqref{eq: chernoff1} to get the result. The derivation is given below: 

Recall that $I(x) = \sup_{\theta} \{x \theta - \ln (\EE[e^{\theta Z_i}])\}$ is the Cram\'er transform of the Bernoulli random variable $Z_i$. It is easy to check that
$$
I(x) = \left\{
\begin{array}{cc}
a \ln \left(\frac{x}{p}\right) + (1-x) \ln\left(\frac{1-x}{1-p}\right) & \mbox{if}~x\in [0,1]\\
+\infty & \mbox{else}
\end{array}
\right.
$$
Also recall that $a\triangleq \left(M^{2/3} + \frac{\nu_i M}{\rho U}\right) / M\!B = M^{-1/3}/B + p$, where $p\triangleq \frac{\nu_i}{\rho BU}$. Since we are considering a large $B$, $a\in [0,1]$ holds. Thus, denoting $\bar{p}=1-p$ for brevity, the exponent of RHS of inequality~\eqref{eq: chernoff1} reads
\begin{eqnarray}
&&-M\!B \cdot I(a) \nonumber\\
&=&-(pM\!B+M^{2/3})\cdot
\ln\left(1+\frac{1}{p M^{1/3}B}\right) \nonumber\\
&&-(\bar{p}M\!B-M^{2/3})\cdot
\ln\left(1-\frac{1}{\bar{p} M^{1/3}B}\right) \nonumber\\
&=& -\frac{pM\!B + M^{2/3}}{p M^{1/3}B}+\frac{pM\!B}{2(pM^{1/3}B)^2} \nonumber\\
&&+\frac{\bar{p}M\!B-M^{2/3}}{\bar{p} M^{1/3}B}+ \frac{\bar{p}M\!B}{2(\bar{p}M^{1/3}B)^2}
+o(M^{1/3}) \nonumber\\
&=&-\frac{M^{1/3}}{2B}\left(\frac{1}{p}+\frac{1}{\bar{p}}\right)+o(M^{1/3}) \nonumber\\
&=&-\frac{M^{1/3}}{2}\left(\frac{\rho U}{\nu_i}+\frac{1}{B(1-\frac{\nu_i}{\rho U})}\right)+o(M^{1/3}) \nonumber\\
&=& -\Theta\left(M^{1/3}\right). \label{eq: decay_rate_bound1}
\end{eqnarray}
With similar steps as above, we can show the exponent exponent of the RHS of inequality \eqref{eq: chernoff2} is also $-\Theta\left(M^{1/3}\right)$. Therefore, inequality \eqref{eq: good_bound} is proved.
\end{document}